\documentclass[runningheads]{llncs}
\usepackage[T1]{fontenc}
\usepackage{graphicx}

\usepackage{hyperref}
\usepackage{color}

\usepackage{mdframed}
\newmdtheoremenv{theo}{Theorem}
\newmdtheoremenv{idef}{Definition}

\usepackage[utf8]{inputenc}
\usepackage[T1]{fontenc}
\usepackage[en-US]{datetime2}

\usepackage{stmaryrd,amssymb, amsfonts,amsmath }
\usepackage{mathdots}
\usepackage{pgf}

\usepackage{thm-restate}
\usepackage{enumitem}
\usepackage{tikz}
\usepackage{pgfplots}
\usepackage{dsfont}

\pgfplotsset{compat = newest}
\usetikzlibrary{calc,arrows,shapes,backgrounds}
\usetikzlibrary{arrows,shapes,decorations,automata,backgrounds,positioning}
\usetikzlibrary{shadows,shadings,shapes.symbols}
\usetikzlibrary{patterns.meta}

\tikzset{
    double color fill/.code 2 args={
        \pgfdeclareverticalshading[%
        tikz@axis@top,tikz@axis@middle,tikz@axis@bottom%
        ]{diagonalfill}{100bp}{%
            color(0bp)=(tikz@axis@bottom);
            color(50bp)=(tikz@axis@bottom);
            color(50bp)=(tikz@axis@middle);
            color(50bp)=(tikz@axis@top);
            color(100bp)=(tikz@axis@top)
        }
        \tikzset{shade, left color=#1, right color=#2, shading=diagonalfill}
    }
}
\newcommand{\tcolA}{yellow!50}

\newcommand{\tcolB}{blue!40}
\newcommand{\tcolBName}{blue}

\tikzstyle{player}=[state,draw,rounded rectangle,align=center]
\tikzstyle{widget}=[draw=red,rectangle, rounded rectangle=10pt,dashed,minimum size=6mm,fill=yellow]
\tikzset{every loop/.style={looseness=7}}

\tikzstyle{player1}=[state,draw,rounded rectangle,align=center]
\tikzstyle{player2}=[state,draw,rectangle,align=center]
\tikzstyle{halfplayer}=[state,semicircle,align=center,anchor=south,rotate=90,scale=0.5,inner sep=0pt, outer sep=0pt]
\tikzstyle{subgraph}=[rectangle,draw,dashed,minimum width=180,minimum height=100,font=\sffamily\Large\bfseries]
\tikzstyle{subgraph2}=[rectangle,draw,dashed,minimum width=95,minimum height=90,font=\sffamily\Large\bfseries]

\usepackage{algorithm}
\usepackage{ stmaryrd }
\usepackage{algpseudocode}

\usepackage{empheq}

\usepackage{multirow}

\newcommand{\MDPtuple}{(\states,\actions,\tprob)}
\newcommand{\tprob}{\ensuremath{\mathbf{P}}} 
\newcommand{\obj}{\ensuremath{\Omega}}

\newcommand{\target}{F}
\newcommand{\formula}{\varphi}
\newcommand{\threshobj}[3]{\mathtt{Pr}_{#1#2}(#3)}
\newcommand{\atomconstr}{atom}
\newcommand{\pareto}[3]{\mathit{Pareto}(#1,#2,#3)}
\newcommand{\achievable}[3]{\mathit{Ach}(#1,#2,#3)}
\newcommand{\closure}[1]{\overline{#1}}
\newcommand{\interior}[1]{\operatorname{int} #1}
\newcommand{\generator}[1]{\operatorname{gen} #1}
\newcommand{\boundary}[1]{\partial #1}

\newcommand{\projSinkState}{\bot}

\newcommand*{\vect}[1]{\mathbf{#1}}

\DeclareMathOperator{\parobj}{\rho}

\newcommand{\univ}{\ensuremath{\mathtt{S}}}

\newcommand{\event}{\ensuremath{{\sf \lozenge}}}

\newcommand*{\satisfy}[3]{#1,#2\models_{#3}}
\newcommand*{\satisfies}[2]{#1\models_{#2}}

\newcommand*{\nsatisfies}[2]{#1\nmodels_{#2}}

\newcommand*{\nmodels}{\ensuremath{\nvDash}}

\newcommand*{\limwedgeone}[1]{{\bigwedge}_{#1}}

\newcommand{\pathname}{\ensuremath{\pi}}

\newcommand{\states}{\ensuremath{S} }

\newcommand{\markovChain}{\ensuremath{\mathcal{M}} }

\newcommand{\MDP}{\ensuremath{\Gamma}}

\newcommand{\MDPtoMC}[2]{\ensuremath{{#1}{[{#2}]}}}
\newcommand{\paths}{\ensuremath{\mathsf{Paths}}}

\newcommand{\nat}{\ensuremath{\mathbb{N}} }

\newcommand{\strat}{\ensuremath{\sigma}}

\newcommand{\NPinter}{\ensuremath{\sf{NP} \cap \sf{coNP}}}
\newcommand{\UPinter}{\ensuremath{\sf{UP} \cap \sf{coUP}}}
\newcommand{\PTIME}{\ensuremath{\sf{PTIME}}}
\newcommand{\EXPTIME}{\ensuremath{\sf{EXPTIME}}}
\newcommand{\PSPACE}{\ensuremath{\sf{PSPACE}}}

\newcommand{\NPTIME}{\ensuremath{\sf{NP}}}

\newcommand{\dists}{\ensuremath{\mathcal{D}} }

\newcommand{\prob}{\ensuremath{\mathsf{Pr}} }

\newcommand{\yes}{\ensuremath{\textsc{Yes}}}

\newcommand{\poly}{\ensuremath{\mathsf{poly}} }

\newcommand{\supp}{\ensuremath{\mathsf{supp}} }

\newcommand{\actions}{\ensuremath{Act}}

\newcommand{\stam}[1]{}

\makeatletter
\tikzset{circle split part fill/.style  args={#1,#2}{%
 alias=tmp@name, 
  postaction={%
    insert path={
     \pgfextra{%
     \pgfpointdiff{\pgfpointanchor{\pgf@node@name}{center}}%
                  {\pgfpointanchor{\pgf@node@name}{east}}%
     \pgfmathsetmacro\insiderad{\pgf@x}
      \fill[#1] (\pgf@node@name.base) ([xshift=-\pgflinewidth]\pgf@node@name.east) arc
                          (0:180:\insiderad-\pgflinewidth)--cycle;
      \fill[#2] (\pgf@node@name.base) ([xshift=\pgflinewidth]\pgf@node@name.west)  arc
                           (180:360:\insiderad-\pgflinewidth)--cycle;            
         }}}}}  
 \makeatother

\colorlet{drouge}{red}
\colorlet{frouge}{red!20!white}
\colorlet{dbleu}{blue}
\colorlet{fbleu}{blue!40!white}
\colorlet{dviolet}{blue!50!red}
\colorlet{fviolet}{blue!50!red!40!white}
\colorlet{dvert}{green!80!black}
\colorlet{fvert}{green!80!black!20!white}
\colorlet{djaune}{yellow!80!black}
\colorlet{fjaune}{yellow!80!black!20!white}
\colorlet{dgris}{white!60!black}
\colorlet{fgris}{white!90!black}
\colorlet{dgrisf}{white!30!black}
\colorlet{fgrisf}{white!70!black}
\colorlet{dorange}{red!50!yellow}
\colorlet{forange}{red!50!yellow!30!white}

\tikzstyle{ptrond}=[draw,circle,minimum height=2mm]
\tikzstyle{ptcarre}=[draw,minimum width=3mm,minimum height=3mm]
\tikzstyle{moyrond}=[draw,circle,minimum height=5mm]
\tikzstyle{moycarre}=[draw,minimum width=4mm,minimum height=4mm]
\tikzstyle{rond}=[draw,circle,minimum height=7mm]
\tikzstyle{carre}=[draw,minimum width=6mm,minimum height=6mm]
\tikzstyle{rouge}=[draw=drouge,fill=frouge]
\tikzstyle{vert}=[draw=dvert,fill=fvert]
\tikzstyle{jaune}=[draw=djaune,fill=fjaune]
\tikzstyle{bleu}=[draw=dbleu,fill=fbleu]
\tikzstyle{violet}=[draw=dviolet,fill=fviolet]
\tikzstyle{orange}=[draw=dorange,fill=forange]
\tikzstyle{gris}=[draw=dgris,fill=fgris]
\tikzstyle{grisf}=[draw=dgrisf,fill=fgrisf]
\tikzstyle{rvert}=[style=rond,style=vert]
\tikzstyle{rrouge}=[style=rond,style=rouge]
\tikzstyle{roundrect}=[draw,rounded rectangle, minimum width=6mm,minimum height=6mm]
\tikzstyle{splitrond}=[draw,circle split,minimum height=7mm,circle split part fill={blue!50,red!50}]
\tikzstyle{splitrondbv}=[draw,circle split,minimum height=7mm,circle split part fill={fbleu,fvert}]
\tikzstyle{splitrondrg}=[draw,circle split,minimum height=7mm,circle split part fill={frouge,fgris}]
\tikzstyle{splitrondbo}=[draw,circle split,minimum height=7mm,circle split part fill={fbleu,forange}]

\def\listof#1{\expandafter\@listof#1+\@end}
\def\@listof#1+#2\@end{\def\@tempa{#1}\ifx\@tempa\@empty\else 
    \langle #1\rangle \fi
  \def\@tempa{#2}\ifx\@tempa\@empty\else,\@listof#2\@end\fi}
\def\stackof#1{\begin{array}{@{}>{\scriptstyle}c@{}}\expandafter\@stackof#1+\@end}
\def\@stackof#1+#2\@end{\def\@tempa{#1}\ifx\@tempa\@empty\else 
    \langle #1\rangle \fi
  \def\@tempa{#2}\ifx\@tempa\@empty\end{array}\else\\[-1mm]\@stackof#2\@end\fi}
  
\usepackage{bwc}

\usepackage{graphicx}
\usepackage[export]{adjustbox}

\tikzset{cross/.style={cross out, draw, 
        minimum size=2*(#1-\pgflinewidth), 
        inner sep=0pt, outer sep=0pt}}

\usepackage{hhline}

\usepackage{cleveref}

\usepackage{booktabs}

\usepackage{nicefrac}

\renewcommand{\mypar}[1]{\subsubsection{#1}}

\begin{document}
\title{Markov Decision Processes with Sure Parity and Multiple Reachability Objectives}
\titlerunning{MDPs with Sure Parity and Multiple Reachability Objectives}
%
\author{Rapha\"el {Berthon}\thanks{Supported by the DFG Grant Nr. 520530521 ``POMPOM''} \and Joost-Pieter {Katoen}$^{\star,\star\star}$\and Tobias {Winkler}\thanks{Supported by the DFG RTG 2236 ``UnRAVeL''}}
\authorrunning{R. Berthon et al.}
%
\institute{RWTH Aachen University, 52062 Aachen, Germany}
%
\maketitle              
\begin{abstract}
This paper considers the problem of finding strategies that satisfy a 
mixture of sure and threshold objectives in Markov decision processes. 
We focus on a single $\omega$-regular objective expressed as parity that 
must be surely met while satisfying $n$ reachability objectives towards sink states 
with some probability thresholds too. 
We consider three variants of the problem: (a) strict and (b) non-strict thresholds on all reachability objectives, and (c) maximizing 
the thresholds with respect to a lexicographic order. 
We show that (a) and (c) can be reduced to solving parity games, and (b) can be solved in \EXPTIME.
Strategy complexities as well as algorithms are provided for all cases.
\end{abstract}
\section{Introduction}
\label{sec:intro}

\emph{Markov decision processes} (MDPs)~\cite{bellman1957markovian,DBLP:books/wi/Puterman94} are prominent models for strategic planning and decision making in face of stochastic uncertainty.
An important, yet intricate, problem is to determine if and how a combination of \emph{multiple} properties, or \emph{objectives}, is realizable in a given MDP.
As objectives may be \emph{conflicting}, it does not suffice to analyze each of them independently~\cite{DBLP:conf/stacs/ChatterjeeMH06,DBLP:journals/jair/RoijersVWD13,DBLP:conf/csl/BaierDK14}. 
Instead, \emph{trade-offs} between the objectives have to be taken into account.
In this paper, we combine objectives of different nature:
\emph{Sure objectives} must be fulfilled on \emph{all} possible executions of the MDP, even on those with probability 0.
Thus, sure objectives do not depend on the exact transition probabilities; in fact, they can be analyzed by replacing the probabilities with an adversary. 
\emph{Threshold objectives}, on the other hand,
have to be satisfied with some probability of at least (or greater than) a given constant. 
Various combinations of sure and threshold objectives have been investigated in prior work~\cite{clemente2015multidimensional,almagor2016minimizing,DBLP:journals/iandc/BruyereFRR17,icalp2017,DBLP:conf/concur/ChatterjeeP19,DBLP:conf/lics/BerthonGR20}.
Here, we focus on MDPs with a single sure $\omega$-regular objective expressed as a \emph{parity} condition together with $n$ \emph{reachability} threshold objectives towards sink states.

\mypar{Running example}
In a game show a contestant plays a gamble to win either a bike, a surfboard, or both prizes.
The gamble is as follows:
The contestant must choose one out of two pairs of 6-sided dice.
Each pair consists of a \textcolor{green!50!black}{green} and a \textcolor{red}{red} die.
The four dice are all different; each of their faces shows either the bike, the board, both bike and board, or the symbol $\circlearrowright$ (``repeat''):
\begin{align*}
    pair_1\qquad & \text{\textcolor{red}{\emph{red}}} \colon 3 \,\times\! \circlearrowright, 1 \times \text{both}, 2 \times \text{bike} && \text{\textcolor{green!50!black}{\emph{green}}} \colon 6 \times \text{board} \\
    pair_2\qquad & \text{\textcolor{red}{\emph{red}}} \colon 3 \,\times\! \circlearrowright, 1 \times \text{both}, 2 \times \text{board} && \text{\textcolor{green!50!black}{\emph{green}}} \colon 2 \times \text{both} ,  2 \times \text{bike} , 2 \times \text{board}
\end{align*}
After committing to a pair of dice, the contestant rolls one die from their pair.
The \textcolor{green!50!black}{green} die immediately ends the gamble with the resulting prize(s).
The \textcolor{red}{red} die either ends the gamble or, in case of $\circlearrowright$, allows the contestant to roll again (the same die or the other one).
However, since the show is broadcast on live TV, there is an additional rule: 
The gamble may not be prolonged indefinitely, i.e., the contestant may try the \textcolor{red}{red} die at most an arbitrary, but \emph{a priori} fixed number of times.
Clearly, optimal strategies depend on how much the contestant prefers one prize over the other.
The MDP in \Cref{fig:mdp-pareto} models this gamble.
Prizes are encoded as reachability of sinks ($\target_1~\hat{=}~\text{bike}, \target_2~\hat{=}~\text{board}$); the additional constraint is a sure parity condition.

\begin{figure}[t]
    \centering
    \adjustbox{max width=\textwidth}{
        \begin{tikzpicture}[on grid, node distance=17.5mm and 20mm]
        \node[player,initial, initial left,initial text=] (s) at (0,0) {$1$} ;
        \node[player,above=of s] (s1) {$1$} ;
        \node[player,below=of s] (s2) {$1$} ;
        \node[player,accepting,right=of s,fill=\tcolA] (r1) {$F_1$} ;
        \node[player,accepting,right=of r1,fill=\tcolB] (r2) {$F_2$} ;
        \node[player,accepting,left=of s,double color fill={\tcolA}{\tcolB}] (r12)  {$F_1\& F_2$} ;
        
        \path[-latex]  (s) edge node[left] { $pair_1,1$} (s1)
        (s) edge node[left] { $pair_2,1$} (s2)
        
        (s1) edge[bend left=20] node[above right, pos=0.7] { $\textcolor{red}{a}, \nicefrac 1 3$} (r1)
        (s1) edge[bend right=20] node[above left] { $\textcolor{red}{a}, \nicefrac 1 6$} (r12)
        (s1) edge[bend left=25] node[above right] { $\textcolor{green!50!black}{b},1$} (r2)
        (s2) edge[bend right=25] node[below right, align=left] { $\textcolor{red}{a}, \nicefrac 1 3$ \\ $\textcolor{green!50!black}{b}, \nicefrac 1 3$} (r2)
        (s2) edge[bend left=20] node[below left, align=left] { $\textcolor{red}{a}, \nicefrac 1 6$ \\ $\textcolor{green!50!black}{b}, \nicefrac 1 3$} (r12)
        (s2) edge[bend right=20] node[below right, pos=0.7] { $\textcolor{green!50!black}{b}, \nicefrac 1 3$} (r1)
        (s1) edge [loop above] node[left, pos = 0.2] { $\textcolor{red}{a}, \nicefrac 1 2$} (s1)
        (s2) edge [loop below] node[left, pos=0.8] { $\textcolor{red}{a}, \nicefrac 1 2$} (s2)
        (r12) edge [loop left] node[above, pos=0.8] {$*,1$} (r12)
        (r1) edge [loop right] node[above, pos=0.2] { $*,1$} (r2)
        (r2) edge [loop right] node[above, pos=0.2] { $*,1$} (r2)
        ;
        \end{tikzpicture}
        \hspace{2em}
        \begin{tikzpicture}[scale=1.0]
            \begin{axis}[
                clip=false,
                axis lines=center,
                xmin = -0.025, xmax = 1.1,
                ymin = -0.025, ymax = 1.1,
                x=40mm, y=40mm,
                xtick={0,0.2,...,1}, ytick={0,0.2,...,1},
                xlabel={$\prob(\event \tikz[baseline=-0.7ex]{ \node[state,accepting,scale=0.33,fill=\tcolA] {} })$}, xlabel style={right},
                ylabel={$\prob(\event \tikz[baseline=-0.7ex]{ \node[state,accepting,scale=0.33,fill=\tcolB] {} })$}, ylabel style={left},
            ]
            \addplot[color=black] coordinates { (1,0.33) (0.33,1)};
            \addplot[color=black] coordinates { (0,1) (0.33,1)};
            \addplot[color=black] coordinates { (1,0.33) (1,0)};
            \node[label={45:{}},circle,fill,inner sep=2pt] at (axis cs:0,1) {};
            \node[label={45:{$(1,\nicefrac 1 3)$}},circle,fill,inner sep=2pt] at (axis cs:1,0.33) {};
            \node[label={45:{$(\nicefrac 2 3,\nicefrac 2 3)$}},circle,fill,inner sep=2pt] at (axis cs:0.67,0.67) {};
            \node[label={45:{$(\nicefrac 1 3,1)$}},circle,fill,inner sep=2pt] at (axis cs:0.33,1) {};
            \node[label={0:{$(\nicefrac 1 2,\nicefrac 16)$}},cross=2pt] at (axis cs:0.5,0.167) {};
            \end{axis}
        \end{tikzpicture}
    }
    \caption{
        An MDP with $\actions = \{\textcolor{red}{a}, \textcolor{green!50!black}{b}, pair_1, pair_2, *\}$, target sets \colorbox{\tcolA}{$\target_1$}, \colorbox{\tcolB}{$\target_2$}, and parity condition $\parobj$ assigning 0 to the sinks and 1 to the non-sinks.
        \emph{Right:} The Pareto frontier.
    }
    \label{fig:mdp-pareto}
\vspace{-15pt}
\end{figure}
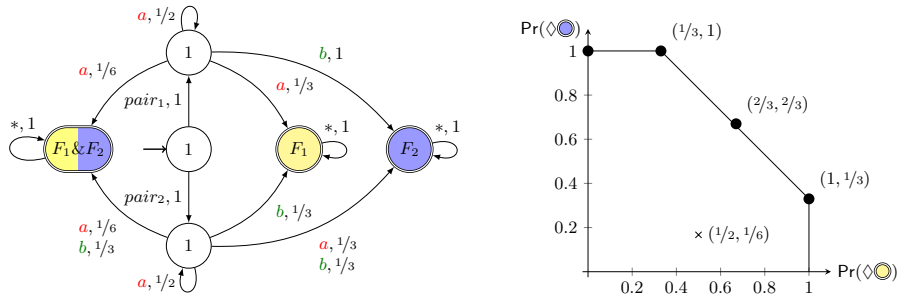

\mypar{This paper}
We study the following three problems:
Given a finite MDP, decide if it is possible to satisfy a sure parity objective and, at the same time, $n$ sink reachability objectives with (a) all \emph{strict}, or (b) all \emph{non-strict} probability thresholds.
In addition, we consider the problem (c) of checking existence of a
\emph{lexicographically} maximal achievable threshold vector w.r.t.\ a given linear order on the reachability objectives.
In all cases, we are also interested in computing witnessing strategies if they exist.
These problems are challenging, both computationally and conceptually.
Already for two reachability objectives (without any sure objective) the set of achievable thresholds ---the \emph{Pareto frontier}--- is a convex polytope with superpolynomially many vertices in general~\cite{etessami2007multi}.
The problems we study are more general than this and add further subtleties:
While it is easy to show that the thresholds in the \emph{interior} of the Pareto frontier are satisfiable with \emph{any} sure parity objective~(\Cref{sec:strict}),
identifying exactly which points on the \emph{boundary} of the frontier are achievable is quite involved (\Cref{sec:nonstrict}).

\begin{table}[t]
\caption{Some existing and our new (in bold) results on multi-objective MDPs.}
\label{tab:comp}
\centering
\resizebox{\textwidth}{!}{%
\setlength{\tabcolsep}{3pt}
\begin{tabular}{l l l l l}
    \toprule
    \textit{Sure objective} & \textit{Probabilistic objective(s)} & \textit{Complexity} & \textit{Memory} & \textit{Reference} \\
    \midrule
    1 sure parity &  -  & $\NPinter$ & finite & \cite{DBLP:journals/tcs/EmersonJS01} \\
    1 sure parity & 1 threshold parity  & $\NPinter$ & infinite & \cite{icalp2017} \\
    1 sure parity & $n$ almost-sure parity & $\NPinter$ & infinite & \cite{DBLP:conf/lics/BerthonGR20} \\
    $n$ sure parity & $n$ almost-sure parity & ${\sf P}^{\sf NP}(=\Delta_2^{\sf P})$ & infinite & \cite{DBLP:conf/lics/BerthonGR20} \\
    1 sure parity & 1 threshold reach & $\NPinter$ & finite & \cite{icalp2017} \\ 
    - & $n$ threshold sink reach & $\PTIME$ & finite & \cite{etessami2007multi} \\
    \textbf{1 sure parity} & \textbf{$n$ strict threshold sink reach} (a) & $\NPinter$ & \textbf{finite} & \textbf{[Thm.\ \ref{thm:approx-int}]} \\
    \textbf{1 sure parity} & \textbf{$n$ threshold  sink reach} (c) & $\EXPTIME$ & \textbf{finite} & \textbf{[Thm.\ \ref{thm:mult-th}]} \\
    - & lexicographic Streett & $\PTIME$ & finite & \cite{chatterjee2023stochastic} \\
    \textbf{1 sure parity} & \textbf{lexicographic  sink reach} (b)             & 
    $\NPinter$ & \textbf{finite} & \textbf{[Thm.\ \ref{thm:mth}]} \\
    \bottomrule
\end{tabular}%
}
\end{table}

\mypar{Contributions} 
Our three main results are summarized in bold in Table~\ref{tab:comp}.
    (a)
    Checking if a sure parity objective and $n$ \emph{strict} sink reachability thresholds are achievable simultaneously is in $\NPinter$ (\Cref{sec:strict}).
    This is done via a reduction to parity games, admitting a quasi-polynomial algorithm~\cite{DBLP:conf/stoc/CaludeJKL017}. 
    
    (c) We propose an algorithm that finds a strategy ensuring a sure parity objective while also maximizing the probability of reaching $n$ sinks w.r.t.\ a \emph{lexicographic} order (\Cref{sec:lex}).
    It relies on a concept we call \emph{projection}, a notion also used in prior work~\cite{etessami2007multi,DBLP:conf/csl/BaierDK14,chatterjee2023stochastic}.
    Our algorithm solves polynomially many parity games in sequence, hence the problem is (again) in $\NPinter$.
    
    (b) We present an algorithm for finding a strategy satisfying a sure parity objective and $n$ sink reachability objectives with \emph{non-strict} thresholds (\Cref{sec:multi}).
    Our algorithm alternates between computing Pareto frontiers, making projections, and pruning states not satisfying the sure parity objective; to our knowledge, this idea is new. 
    Its time complexity is exponential in the size of the MDP, as it relies on computing exact Pareto frontiers.

We also treat strategy complexity for each case.
Our results are a further step towards a solution for general combinations of sure and probabilistic objectives.

\mypar{Related work}
Previous research~\cite{DBLP:conf/lics/BerthonGR20} on mixtures of sure and probabilistic objectives focused on \emph{qualitative} thresholds, i.e., ${>}0$ and ${=}1$. 
Here, we also allow \emph{quantitative} thresholds strictly between $0$ and $1$.
We rely on some results of~\cite{icalp2017} that studied combining one sure parity and \emph{one threshold parity} objective.
This problem was shown to be in $\NPinter$, via a reduction to parity games with weights that can be solved in quasi-polynomial time~\cite{DBLP:journals/lmcs/ScheweWZ19}. 
The main difference to~\cite{icalp2017} is that we consider \emph{multiple reachability} threshold objectives.
The setting of one sure parity and one almost-sure parity has been further studied in~\cite{DBLP:conf/concur/ChatterjeeP19} where it was shown that the restriction to finite memory strategies is still in $\NPinter$ for MDPs, and co$\NPTIME$-complete for stochastic games.

The seminal paper~\cite{etessami2007multi} shows that computing the Pareto frontier for a mixture of either reachability or $\omega$-regular objectives can be reduced to solving linear programs.
An efficient technique, value iteration, is exploited by tools such as PRISM~\cite{DBLP:conf/cav/KwiatkowskaNP11}, Storm~\cite{DBLP:conf/cav/DehnertJK017}, and MultiGain~\cite{DBLP:conf/tacas/BrazdilCFK15}.
The work~\cite{chatterjee2015unifying} considers MDPs with two different kinds of stochastic mean-payoff objectives, and supports computing the Pareto frontier.
\emph{Percentile queries}, multiple threshold constraints that must each be satisfied with some probability, were studied in~\cite{percentile2017}.
In~\cite{bouyer2018multi,hartmanns2020multi}, multiple reachability conditions associated to the expected or accumulated cost to reach a target are considered.

\emph{Lexicographic optimization} is a widely employed principle in multi-objective decision making~\cite{DBLP:conf/socs/MiuraWZ22,DBLP:conf/aaai/WrayZM15}. 
The idea is that a strategy should prioritize a primary objective while still doing best possible for a secondary objective, etc.
The work of \cite{chatterjee2023stochastic} imposes a lexicographic order on multiple, possibly conflicting, reachability, safety and $\omega$-regular objectives.
Reinforcement learning with lexicographic $\omega$-regular conditions is studied in~\cite{DBLP:conf/tacas/HahnPSS0W20}.
To the best of our knowledge, lexicographic optimization in MDPs together with a sure condition has not been studied yet.

Other approaches have been considered. Combinations of parity and mean-payoff~\cite{almagor2016minimizing}, and parity and weighted games~\cite{DBLP:journals/lmcs/ScheweWZ19} have been studied in prior work.
An alternative way to combining objectives is \emph{strategy logic} (SL)~\cite{chatterjee2010strategy}, an extension of CTL that can express formulas involving the change of strategies. 
A probabilistic SL has been defined in~\cite{aminof2019probabilistic}. 

\section{Preliminaries}
\label{sec:prelims}

We write $\mathbb{N} = \{1,2,3,\ldots\}$ and $\mathbb{N}_0 = \mathbb{N} \cup \{0\}$.
For $n \in \mathbb{N}$, we let $[n] = \{1,\ldots,n\}$, and $[n]_0 = [n] \cup \{0\}$.
Vectors $\vect{v} \in \mathbb{R}^n$ are written in bold.
For $\vect{v}\in [0,1]^n$ and $i\in [n]$, we denote by $v_i$ the $i$-th component of $\vect{v}$.
Given $\vect{v}, \vect u \in \mathbb{R}^n$, their \emph{dot product} is defined as $\vect v \cdot \vect u = \sum_{i \in [n]} v_i \cdot u_i$.
The symbol $\vect{e}_i \in \{0,1\}^n$ is the unit vector where the $i$-th component is $1$ and all others are $0$.
The componentwise order on $\mathbb{R}^n$ is denoted with $\leq$.
Given a finite set $A$, a \textit{(probability) distribution} on $A$ is a function $f \colon A \to [0,1]$ such that $\sum_{a\in A} f(a) = 1$. 
$\dists(A)$ denotes the set of all distributions on $A$.
We define the \emph{support} $\supp(f) = \{a \in A \mid f(a) > 0 \}$. 

\subsection{MDPs, Strategies, and Objectives}

A \emph{Markov decision process} (MDP) is a tuple $\MDP = \MDPtuple$ where $\states \neq \emptyset$ is a countable set of \emph{states}, $\actions \neq \emptyset$ is a finite set of \emph{actions}, and $\tprob \colon \states \times \actions \times \states \to [0,1]$ is a \emph{transition probability function} satisfying $\sum_{s' \in \states} \tprob(s,a,s') \in \{0,1\}$ for all $s \in \states, a \in \actions$ .
If the sum is $1$ for a state-action pair $s,a$, then $a$ is \emph{enabled} at $s$.
We write $\actions(s)$ for the set of all actions enabled at $s$, and require that $\actions(s) \neq \emptyset$.
An MDP $\MDP$ is called \emph{finite} if $\states$ is finite.
A state $s \in \states$ is called a \emph{sink} if for all $a \in \actions(s)$ we have $\tprob(s,a,s) = 1$.
For technical convenience we assume $|\actions(s)| = 1$ for all sinks $s \in \states$.
Note that we may consider the same MDP with different initial states since the latter is not fixed in our definition.
See \Cref{fig:mdp-pareto} for a finite example MDP.

A (discrete-time) \emph{Markov chain} (MC) is an MDP with $|\actions| = 1$.
We omit $\actions$ from the definition of MC and just write $\markovChain = (\states, \tprob)$.
We also identify $\tprob$ with a function of type $\states \times \states \to [0,1]$. 

A \emph{strategy} for an MDP $\MDP$ is a state machine $\strat = (Q,q_\iota,\delta, o)$ where $Q$ is a countable set of memory modes, $q_\iota \in Q$ is the initial mode, $\delta \colon Q \times \states \to Q$ is a transition function, and $o \colon Q \times \states \to \dists(\actions)$ is an output function with $\supp(o(q,s)) \subseteq \actions(s)$ for all $q \in Q, s \in \states$.
$\strat$ is called \emph{finite-memory} if $|Q| < \infty$, \emph{memoryless} if $|Q| = 1$, and \emph{deterministic} if $|\supp(o(q,s))| = 1$ for all $q \in Q, s \in S$.

Given an MDP $\MDP = \MDPtuple$ with strategy $\strat = (Q,q_\iota,\delta, o)$, we define the \emph{induced MC} $\MDPtoMC{\MDP}{\strat} = (\states \times Q, \tprob^\strat)$ where $\tprob^\strat((s,q), (s',q')) = \tprob(s, o(q,s),s')$ if $\delta(q,s) = q'$, and otherwise $\tprob^\strat((s,q), (s',q')) = 0$.
In the following, we only consider finite MDPs,  but when considering an infinite-memory strategy, the resulting MC is countably infinite.
In the context of algorithms, we always assume that the probabilities in the given MDPs, strategies, and probability thresholds are \emph{rational} numbers encoded as numerator-denominator pairs in binary.

Given an MC $\markovChain = (\states, \tprob)$ with a distinguished initial state $s \in \states$, we consider the $\sigma$-algebra $\mathcal F$ generated by the \emph{cylinder sets} $\{\pathname \states^\omega \mid \pathname \in \states^*\}$ and the associated probability measure $\prob_s^\markovChain \colon \mathcal F \to [0,1]$ which is uniquely defined by requiring that for all $\pathname = s_0\ldots s_k \in S^+$, $k \geq 0$, we have $\prob_s^\markovChain(\pathname \states^\omega) = \prod_{i=0}^{k-1} \tprob(s_i,s_{i+1})$ if $s_0 = s$, and $\prob_s^\markovChain(\pathname \states^\omega) = 0$ if $s_0 \neq s$.
See, e.g.~\cite[Chapter 10]{baier2008principles} for more details.
The sets in $\mathcal F$ are called \emph{measurable}.
Further, we define $\paths^\markovChain_s = \{s_0 s_1 \ldots \in S^\omega \mid s_0 = s \land \forall i \geq 0 \colon \tprob(s_i, s_{i+1}) > 0\}$.

An \emph{objective} for an MDP $\MDP$ is a measurable\footnote{Measurability is actually only important for probabilistic objectives, not for sure objectives. However, all concrete objectives considered in this paper are measurable.} set of paths $\Omega\subseteq\states^\omega$.
A \emph{reachability} objective for $\MDP$ is of the form $\event\target = \{\pathname \in \states^\omega \mid \exists k \geq 0 \colon \pathname(k) \in \target\}$ where $\target \subseteq \states$.
\emph{Bounded reachability} objectives have the form $\event^{\leq B}\target = \{\pathname \in \states^\omega \mid \exists k \in [B]_0 \colon \pathname(k) \in \target\}$, for some $B \in \nat_0$.
A \emph{parity} objective for $\MDP$ is defined via a \emph{priority function} $\parobj \colon \states \to [k]_0$, where $k \in \mathbb{N}_0$. 
For $\pathname \in \states^{\omega}$, let $\inf(\pathname) = \{ s \in \states \mid \forall i\geq 0 \colon  \exists j\geq i \colon \pathname(j) = s\}$ be the set of states visited infinitely often on $\pathname$.
Then the parity objective defined by $\parobj$ is $\{\pathname \in \states^\omega \mid \max \parobj(\inf(\pathname)) \text{ is even}\}$.
In the following, we identify the function $\parobj$ with the objective it defines.

\subsection{Multi-Objectives and Pareto Frontiers}
\label{sec:moprelims}

A \emph{multi-objective (MO) formula} for MDP $\MDP$ is a syntactic object $\formula = \bigwedge_{i =1}^n \atomconstr_i$ with $\atomconstr_i \in \{\univ(\obj), \threshobj{\sim}{p}{\obj} \mid \obj \text{ an objective for } \MDP, \sim \: \in \{>, \geq\}, p \in [0,1] \}$.
An MC $\markovChain$ with a distinguished initial state $s$ satisfies a \emph{threshold} constraint $\threshobj{\sim}{p}{\obj}$ if $\prob^\markovChain_s(\obj) \sim p$, a \emph{sure} constraint $\univ(\obj)$ if $\paths_s^{\markovChain} \subseteq \obj$, and it satisfies the formula $\formula$ (in symbols: $\satisfies{s}{\markovChain} \formula$) if it satisfies $\atomconstr_i$ for all $i \in [n]$.
For an MDP $\MDP$ with strategy $\strat$ we write $\satisfy{s}{\strat}{\MDP} \formula$ if $\satisfies{s}{\MDPtoMC{\MDP}{\strat}} \formula$,
and $\satisfies{s}{\MDP} \formula$ if $\satisfy{s}{\strat}{\MDP} \formula$ for \emph{some} strategy $\strat$ for $\MDP$.
In this paper, we only consider formulas of the form $\univ(\parobj) \land \threshobj{\sim}{p_1}{\event\target_1} \land \ldots \land \threshobj{\sim}{p_n}{\event\target_n}$, i.e., conjunctions of one sure parity and $n \geq 1$ reachability objectives, either all with strict or non-strict thresholds.

We now define Pareto frontiers.
Given an MO formula $\formula$ for MDP $\MDP$ containing $n$ threshold constraints $\threshobj{\sim}{p_1}{\obj_1}, \ldots, \threshobj{\sim}{p_n}{\obj_n}$ (in this order), we write $\formula(p_1,\ldots,p_n) = \formula(\vect p)$ to emphasize the dependency of $\formula$ on the threshold vector $\vect p \in [0,1]^n$.
We also write $\formula(\vect x)$ without further qualifying $\vect x$ to indicate that the thresholds are \emph{variables} $\vect x = (x_1,\ldots,x_n)$.
We define the set of \emph{achievable threshold vectors} as $\achievable{\MDP}{s}{\formula(\vect x)} = \{\vect{p} \in [0,1]^n \mid \satisfies{s}{\MDP}\formula(\vect{p})\} \subseteq [0,1]^n$.
Note that $\achievable{\MDP}{s}{\formula(\vect x)}$ is \emph{downward-closed} as $\sim \: \in \{\geq, >\}$, and \emph{convex} since a convex combination $c \cdot \vect p + (1{-}c) \cdot \vect{p'}$, $c \in (0,1)$, of achievable threshold vectors $\vect p, \vect{p'}$ is achieved by a strategy that plays the strategy for $\vect p$ with probability $c$ and the one for $\vect{p'}$ with probability $1{-}c$, see, e.g.~\cite{etessami2007multi}.
Given a set of vectors $X$, we define $\generator(X)$, the subspace generated by $X$, as the intersection of all subspaces of $\mathbb{R}^n$ containing $X$. It is the smallest subspace containing $X$.
For the next definition recall that the \emph{boundary} $\boundary{X}$ of a set $X \subseteq [0,1]^n$ is defined as $\closure{X} \setminus \interior{X}$, where $\closure{X}$ is the closure of $X$ and $\interior{X}$ is the interior, i.e., the largest open subset of $X$.

\begin{definition}[Pareto frontier]
    \label{def:pareto}
    Let $\MDP = \MDPtuple$ be an MDP, $\formula(\vect x)$ an MO formula for $\MDP$, and $s\in S$.
    We define $\pareto{\MDP}{s}{\varphi(\vect x)} = \boundary{\achievable{\MDP}{s}{\formula(\vect x)}}$.
\end{definition}
The above definition is similar to the one from~\cite{DBLP:conf/lics/AshokCKWW20}; other authors define the Pareto frontier in a slightly different way, e.g., as the $\leq$-maxima of $\achievable{\MDP}{s}{\formula(\vect x)}$~\cite{DBLP:conf/atva/ForejtKP12}.

The Pareto frontier is the boundary of a convex polytope of dimension at most $n$~\cite{etessami2007multi}.
Such a polytope $P$ has faces of lower dimension, from $0$ (a vertex) to $n-1$. These faces are defined as follows: given a hyperplane $H$ intersecting $P$, the polytope $H\cap P$ is a face of $P$ iff $P$ lies fully on one of the two closed half-spaces defined by $H$.  When considering the polytope associated to a Pareto frontier (and by generalization the Pareto frontier itself), we can freely separate points between those strictly in the interior, those on the border, and those in the exterior. 
In what follows, we only consider faces defined by intersection with hyperplanes whose normal vectors only have non-negative components. More on convex polytopes can be found in~\cite{grunbaum1967convex}. 

\begin{example}
    \label{ex:run}
    Consider the MDP $\MDP$ in \Cref{fig:mdp-pareto} on \cpageref{fig:mdp-pareto} and the formula $\varphi(\vect x) = \threshobj{\geq}{x_1}{\event F_1} \land \threshobj{\geq}{x_2}{\event F_2}$ (we ignore parity).
    Let $s$ be the marked initial state.
    \begin{itemize}
        \item $\vect p_1 = (1, \nicefrac 1 3)$ is achievable from $s$
                by choosing $pair_1$, 
        and then playing $\textcolor{red}{a}$ repeatedly to reach $F_1$ with $\nicefrac 2 3$, and both $F_1$ and $F_2$ with probability $\nicefrac 1 3$.
        \item $\vect p_2 = (\nicefrac 1 3, 1)$ is achievable from $s$ by choosing $pair_2$, and then $\textcolor{red}{a}$ repeatedly to reach $F_2$ with $\nicefrac 2 3$ and both $F_1$ and $F_2$ with probability $\nicefrac 1 3$.
        \item As mentioned earlier, a convex combination of two achievable points is achievable by following one of the two strategies with suitable probabilities.
        However, in this specific example, the vector $\vect p_3 = (\nicefrac 2 3, \nicefrac 2 3) = 0.5\cdot \vect p_1 + 0.5\cdot \vect p_2$ is achievable with a \emph{deterministic} strategy as well:
        First choose $pair_2$ in $s$ and then $\textcolor{green!50!black}{b}$ to reach state $F_1$ with $\nicefrac 1 3$, state $F_2$ with $\nicefrac 1 3$, and both $F_1$ and $F_2$ with probability $\nicefrac 1 3$. These points will be relevant in Sections~\ref{sec:lex} and~\ref{sec:nonstrict}.
        \item The above strategies are all \emph{Pareto-optimal}, but there are also sub-optimal strategies, e.g., choosing $pair_1$ in $s$ and then $\textcolor{green!50!black}{b}$ leads to reaching $\target_1$ and $\target_2$ with probability $0$ and $1$, respectively. This is sub-optimal as $(0,1) \leq \vect p_2$.
    \end{itemize}
\end{example}

We consider \emph{clean} MDPs throughout the rest of the paper:

\begin{definition}[Clean MDP]
    \label{def:clean}
    Let $\MDP = \MDPtuple$ be an MDP.
    $\MDP$ is \emph{clean} ...
    \begin{itemize}
        \item ... w.r.t.\ a parity objective $\parobj \colon \states \to [k]_0$ if for all $s\in S$, we have $\satisfies{s}{\MDP}\univ(\parobj)$, i.e., $\parobj$ is \emph{surely} satisfiable from every state $s$.
        \item ... w.r.t.\ target sets $\target_1,\ldots,\target_n \subseteq \states$ if
        for all $s\in S$, we have $\satisfies{s}{\MDP} \threshobj{\geq}{1}{\event\target}$, where $\target = \bigcup_{i=1}^n \target_i$, and
        every state in $F$ is a sink.
    \end{itemize}
\end{definition}

\begin{example}
    The MDP from \Cref{fig:mdp-pareto} is clean w.r.t.\ $\parobj$ because from every state, there is a strategy that surely reaches a sink with priority $0$.
    For instance, from the topmost state, the rightmost sink is reachable by playing action $b$.
    The MDP is also clean w.r.t.\ $\target_1,\target_2$ because from every state there exists a strategy reaching $\target = \target_1 \cup \target_2$ with probability one, and because $\target$ contains sink states only.
\end{example}

Some remarks about clean MDPs are in order:
    (i)
    One can \emph{clean} an MDP w.r.t.\ parity by identifying and removing states that violate $\univ(\parobj)$.
    The latter can be done by solving the 2-player deterministic parity game obtained by replacing the randomness in the MDP by an antagonistic player.
    Note that deciding the winner in a parity game (and hence checking if $\satisfies{s}{\MDP}\univ(\parobj)$ holds for state $s$) is in \NPinter~\cite{DBLP:conf/stoc/CaludeJKL017}, and even in \UPinter~\cite{DBLP:journals/ipl/Jurdzinski98}, but is not known to be in $\PTIME$.
        (ii)
    Reachability towards \emph{sinks} only is a more severe restriction.
    We make it because simultaneous almost-sure reachability of $n$ general target sets in an MDP is already \PSPACE-complete~\cite{percentile2017}.
    Intuitively, this is because strategies have to remember which targets were already seen.
    Contrarily, sink reachability often admits more practical complexities as shown in \Cref{sec:strict,sec:lex} (also see, e.g.,~\cite{chatterjee2023stochastic} and \cite{DBLP:journals/corr/abs-2109-08317}) and is still of practical interest.
    We leave a solution for general reachability for future work, since it would likely further improvements on the techniques we introduce.

\subsection{Existing Results}
\label{app:existing_results}
We recall the following known results:

\begin{theorem}[Known results]
    \label{thm:known_results}
    Let $\MDP = \MDPtuple$ be a finite MDP with target set $\target \subseteq \states$.
    \begin{enumerate}
        \item\label{it:opt_reach}\label{thm:opti_reach} \textnormal{\textsf{({Optimal reachability; see, e.g.~\cite[Lem.~10.102]{baier2008principles} or~\cite[Sec.~28.4.1]{clarke2018handbook}}).}}
        There exists a memoryless and deterministic strategy $\sigma^{\ast}$ for $\MDP$ such that for all $s \in \states$ we have
        $\prob^{\MDPtoMC{\MDP}{\sigma^{\ast}}}_s(\event \target)
        =
        \sup_{\sigma} \prob^{\MDPtoMC{\MDP}{\sigma}}_s(\event \target)
        =:
        v^{\ast}_{s}$, where the supremum ranges over all strategies for $\MDP$.
        Moreover, both $\sigma^{\ast}$ and the probabilities $(v^{\ast}_{s})_{s \in S}$ are computable in polynomial time.

        \item\label{it:sure_parity}\textnormal{\textsf{({Sure parity~\cite[Thm.~2.1]{DBLP:conf/stoc/CaludeJKL017}}).}}
        Let $\parobj\colon \states \rightarrow [k]_0$ be a priority function.
        Then deciding $\satisfies{s}{\MDP} \univ(\parobj)$ for a state $s \in \states$ is in ${\sf NP} \cap {\sf coNP}$,
        and there exist memoryless deterministic witness strategies for the \yes-instances.
        \item\label{it:sure_parity_and_reach}\label{thm:reach-p}\label{thm:comp-sas} \textnormal{\textsf{({Sure parity and threshold reachability~\cite[Thm.~4]{icalp2017}}).}}
        Let $\parobj\colon \states \rightarrow [k]_0$ be a priority function, and $p \in [0, 1]$ a probability threshold.
        Then deciding $\satisfies{s}{\MDP} \univ(\parobj) \wedge \threshobj{\sim}{p}{\event \target}$ for a state $s \in \states$ and $\sim\, \in \{>, \geq\}$ is in ${\sf NP} \cap {\sf coNP}$,
        and there exist polynomial memory (possibly randomized) witness strategies for the \yes-instances.
        \item\label{it:mo_reach} \textnormal{\textsf{(Multi-objective reachability~\cite[Cor.~3.5]{etessami2007multi}).}}
        Suppose that $\MDP$ is clean w.r.t.\ $\\ \target_1,\ldots,\target_n \subseteq \states$ and let $\sim~\in \{\geq, >\}$.
        Then deciding $\satisfies{s}{\MDP} \bigwedge_{i = 1}^n \threshobj{\sim}{p_i}{\event \target_i}$ for a state $s \in \states$ and thresholds $p_1,\ldots,p_n \in [0,1]$ is in $\PTIME$, and memoryless (but possibly randomized) witness strategies for the \yes-instances can be constructed in polynomial time.
    \end{enumerate}
\end{theorem}

\section{Sure Parity and $n$ Strict Reachability Thresholds}
\label{sec:strict}
\label{sec:thres-r}

We study MO formulas of the form $\univ(\parobj) \wedge \bigwedge_{i = 1}^n \threshobj{>}{p_i}{\event \target_i}$ in this section, i.e., with \emph{strict} thresholds only.
Non-strict thresholds are more involved, see~\Cref{sec:nonstrict}.
We start by stating the main result of this section.
Note that it is formulated for MDPs that are clean w.r.t.\ the parity objective $\parobj$ and the target sets $\target_i$.
The assumption of being clean w.r.t.\ parity can be dropped, but this incurs the additional complexity of solving a parity game (see \Cref{def:clean} and subsequent remarks), and hence leads to an $\NPinter$ complexity bound on the associated decision problem.

\vspace{15pt}
\begin{restatable}{theorem}{approxInt}
    \label{thm:approx-int}
    \begin{mdframed}[innertopmargin=+0.5em,innerbottommargin=+0.5em]
        Let $\MDP$ be a clean MDP w.r.t.\  parity objective $\parobj$ and target sets $\target_1,\ldots,\target_n \subseteq \states$.
        Further, let $\vect{p} \in [0,1]^n$, and let $s \in S$ be a state.
        Then:
        \begin{enumerate}
            \item\label{it:approxIntFirst} The decision problem $\exists\strat \colon \satisfy{s}{\strat}{\MDP} \univ(\parobj) \wedge \bigwedge_{i = 1}^n \threshobj{>}{p_i}{\event \target_i}$ is in $\PTIME$.
            \item\label{it:approxIntSecond} A witness strategy $\strat$ using at most $2^{\poly(|\MDP| + nD)}$ memory, where $D$ is the bit-complexity of the rational numbers in $\vect p$, can be effectively constructed for the \yes-instances.
        \end{enumerate}
    \end{mdframed}
\end{restatable}

\begin{example}
    Reconsider the MDP $\MDP$ from \Cref{fig:mdp-pareto} with initial state $s$, yellow target $F_1$ and blue target $F_2$.
    To surely satisfy $\parobj$, we must visit non-sink states only finitely often (this is a co-B\"uchi condition).
    We show that $\satisfies{s}{\MDP}\univ(\parobj)\wedge \threshobj{>}{\nicefrac 1 2}{\event F_1}\wedge \threshobj{>}{\nicefrac 1 6}{\event F_2}$, achieving a value strictly greater than $\vect p_4 = (\nicefrac 1 2,\nicefrac 1 6 )$ which is strictly inside the Pareto frontier. To achieve this objective, we take the following strategy: we first play action $pair_1$, then 
$\textcolor{red}{a}$ twice. By doing so, we have probability $\nicefrac 1 4 $ of reaching the leftmost state (contained in both $F_1$ and $F_2$), and probability $\nicefrac 1 2 $ of reaching the state only fulfilling $F_1$. If we now play action $\textcolor{green!50!black}{b}$, we satisfy condition $\parobj$ surely.
    We end up reaching $F_1$ with probability $\nicefrac 3 4 > \nicefrac 1 2$ and reaching $F_2$ with probability $\nicefrac 1 2 >\nicefrac 1 6 $.
    Note that in general, thresholds that are achievable with strict inequalities are located strictly inside the Pareto frontier.
\end{example}

\Cref{thm:approx-int} relies on the following two lemmas.

\begin{restatable}{lemma}{moReachBound}
    \label{lem:moReachBound}
    Suppose that $\MDP$ is clean w.r.t.\ $\target_1,\ldots,\target_n \subseteq \states$,
    and $s \in \states$ is such that $\satisfies{s}{\MDP} \bigwedge_{i=1}^n \threshobj{>}{p_i}{\event \target_i}$, where all the $p_i$ are probabilities representable as fractions of integers with bit-size at most $D$.
    Then there exists a memoryless (but possibly randomized) strategy $\strat$ and $B \leq 2^{\poly(|\MDP| + nD)}$ such that
    \[
        \satisfy{s}{\strat}{\MDP} \bigwedge_{i=1}^n \threshobj{>}{p_i}{\event^{\leq B} \target_i}
        ~.
    \]
\end{restatable}

\begin{restatable}{lemma}{sameAch}
    \label{lem:same-ach}
    Under the assumptions of \Cref{thm:approx-int} and $\formula(\vect x) = \bigwedge_{i = 1}^n \threshobj{>}{x_i}{\event \target_i}$,
    \[
        \achievable{\MDP}{s}{\formula(\vect x)} = \achievable{\MDP}{s}{\univ(\parobj) \wedge \formula(\vect x)}
        ~,
    \]
    i.e., 
    the Pareto frontiers of $\formula(\vect x)$ and \mbox{$\univ(\parobj) \wedge \formula(\vect x)$ coincide}.
\end{restatable}

\subsection{Proof of Lemma~\ref{lem:moReachBound}}
\label{app:moReachBound}

We use the following auxiliary result to prove \Cref{lem:moReachBound}:

\begin{restatable}{lemma}{mcbound}
    \label{lem:mc-bound}
    Consider a finite MC $\markovChain = (\states, \tprob)$ with a state $s \in \states$, target set $\target \subseteq \states$, and a rational probability $p = q/r$ where $q \geq 0, r > 0$ are integers with bit-size\footnote{A positive integer $n$ has bit-size $D = \lfloor \log_2 n \rfloor + 1$. The integer $0$ has bit-size $1$.} at most $D$.
    Then if $\satisfies{s}{\markovChain} \threshobj{>}{p}{\event \target}$, we also have that $\satisfies{s}{\markovChain} \threshobj{>}{p}{\event^{\leq B} \target}$, where $B$ can be chosen such that $B \leq 2^{\poly(|\markovChain|) + D}$ for a fixed polynomial $\poly$.
\end{restatable}
\begin{proof}
    We give a proof based on Markov's inequality.
    Recall that Markov's inequality states that for a non-negative random variable $X$ and all $a >0$, we have $\prob(X \geq a) \leq \frac{E[X]}{a}$, where $E[X]$ denotes the expected value of $X$.
    Now let $X_\target \colon S^\omega \to \nat_0 \cup \{\infty\}$ be the random variable $X_\target(s_0s_1\ldots) = \min \{i \mid s_i \in \target\}$, with the convention that $\min \emptyset = \infty$.
    Intuitively, $X_F$ is the number of steps until $\target$ is visited for the first time, or $\infty$ if $F$ is not reached at all.
    We write $E_s[X_\target]$ for the expected value of $X_\target$ w.r.t.\ probability measure $\prob_s^\markovChain$.
    Applying Markov's inequality directly to $X_\target$ is not useful because $E_s[X_\target] = \infty$ is possible.
    However, since $\markovChain$ is finite, $E_s[X_\target] = \infty$ if and only if $\prob(\event \target) < 1$, and moreover, $\prob(X_\target = \infty) = 1 - \prob(\event \target)$.
    Now consider the $\nat_0$-valued random variable $\tilde{X}_\target$ which is defined in the same way as $X_\target$ except that $\tilde{X}_\target$ is $0$ if $X_\target$ is $\infty$.
    With these definitions, it holds for all integers $B \geq 0$:
    \begin{align*}
    \prob_s(\event^{\leq B} \target) ~=~& 1 - \prob_s(X_\target \geq B+1) \\
    ~=~& 1 - \prob_s(\tilde{X}_\target \geq B+1) - \prob_s(X_\target = \infty) \tag{by definition of $\tilde{X}_F$}\\
    ~\geq~& 1 - \frac{E_s[\tilde{X}_\target]}{B+1} - (1 - \prob_s(\event \target)) \tag{Markov's inequality}\\
    ~=~& \prob_s(\event \target) - \frac{E_s[\tilde{X}_\target]}{B+1} \tag{simplifying}
    \end{align*}
    It follows that if $B > \frac{E_s[\tilde{X}_\target]}{\prob_s(\event \target) - p}$ then $\prob_s(\event^{\leq B} \target) > p$.
    
    It only remains to bound the quantity $\frac{E_s[\tilde{X}_\target]}{\prob_s(\event \target) - p}$.
    \begin{itemize}
        \item We consider the denominator first.
        $\prob_s(\event \target)$ is a rational number computable in polynomial time in $|\markovChain|$ (by solving a linear equation system; see, e.g.,~\cite[Chapter 10]{baier2008principles}), thus $\prob_s(\event \target)$ can be represented as a fraction of integers where both the numerator and the denominator have at most $\mathcal O(\poly(|\markovChain|))$ many bits.
        As a consequence, the number $\prob_s(\event \target) - p$ can be represented as a fraction of integers with bit-size at most $\mathcal O(\poly(|\markovChain|) + D)$.
        \item Similarly, the expected value of $\tilde{X}_\target$ can also be expressed as the unique solution of a linear equation system with $|\states|$ many equations and rational coefficients whose bit-size is at most polynomial in $|\markovChain|$.
        We state this equation system (which appears to be somewhat less well known; see, e.g., equation (9) in~\cite{DBLP:conf/lics/WinklerK23}) for the sake of completeness:
        For all $s \in \states$ we have a variable $x_s$ and the equation
        \begin{align*}
        x_s = \begin{cases}
        \prob_s(\event \target) + \sum_{s' \in \states} \tprob(s,s') \cdot x_{s'} &\text{ if } s \notin \target \land \prob_s(\event \target) > 0 \\
        0 &\text{ else}
        \end{cases}
        \end{align*}
        The unique solution of the above equation system is $x_s = E_s[\tilde{X}_\target]$.
        Since the coefficients in the equation system have polynomial bit size (see previous point), the same applies to the expected value $E_s[\tilde{X}_\target]$.
        \item It follows that $\frac{E_s[\tilde{X}_\target]}{\prob_s(\event \target) - p}$ can be represented as fraction of integers where the numerator has bit-size at most $\mathcal O(\poly(|\markovChain|) + D)$ for a fixed polynomial $\poly$.
        Hence it suffices to choose $B = 2^{\poly(|\markovChain|) + D}$.
        This completes the proof.
    \end{itemize}
\end{proof}

\begin{proof}[of \Cref{lem:moReachBound}]
    Since $\satisfies{s}{\MDP} \bigwedge_{i=1}^n \threshobj{>}{p_i}{\event \target_i}$, \Cref{thm:known_results}(\ref{it:mo_reach}) asserts existence of a memoryless (but possibly randomized) strategy $\strat$ that can be constructed in time at most $\mathcal O(\poly(|\MDP| + nD))$, i.e., the probabilities in $\strat$ are representable as fractions of integers with bit-size at most $\mathcal O(\poly(|\MDP| + nD))$.
    Now consider the induced MC $\MDPtoMC{\MDP}{\strat}$.
    We have that $|\MDPtoMC{\MDP}{\strat}| \in \mathcal O(\poly(|\MDP| + nD))$.
    Applying \Cref{lem:mc-bound} we find that each individual $\target_i$ is reached with probability strictly greater than $p_i$ after at most $2^{\poly(|\MDP| + nD)}$ steps.
\end{proof}

\subsection{Proof of Lemma~\ref{lem:same-ach}}
\label{app:sameach}

\sameAch*

\newcommand{\prstrat}{\strat_{\mathtt{Pr}}}
\newcommand{\univstrat}{\strat_{\mathtt{S}}}
\begin{proof}
    We have to show that for all $\vect p \in [0,1]^n$, if $\satisfies{s}{\MDP} \formula(\vect p)$ then $\satisfies{s}{\MDP} \univ(\parobj) \wedge \formula(\vect p)$.
    Suppose that $ \satisfies{s}{\MDP} \formula(\vect p)$.
    By \Cref{lem:moReachBound} there exists a memoryless (but possibly randomized) strategy $\prstrat$ such that $\satisfy{s}{\prstrat}{\MDP} \bigwedge_{i=1}^n \threshobj{>}{p_i}{\event^{\leq B} \target_i}$, and thus in particular $\satisfy{s}{\prstrat}{\MDP} \formula(\vect p)$.
    Moreover, by the assumption that $\MDP$ is clean w.r.t.\ $\parobj$ there exists a memoryless deterministic strategy $\univstrat$ such that $\satisfy{s'}{\univstrat}{\MDP} \univ(\parobj)$ for all $s' \in \states$.
    We now define a finite-memory strategy $\strat$ as follows:
    Play $\prstrat$ for $B$ steps, then play according to $\univstrat$ ad infinitum.
    $\strat$ uses finite-memory because it only has to ``count'' to $B$ and because both $\prstrat$ and $\univstrat$ are memoryless.
    Then clearly $\satisfy{s}{\strat}{\MDP} \formula(\vect p)$, and moreover $\satisfy{s}{\strat}{\MDP} \univ(\parobj)$ because parity is prefix-independent and $\strat$ switches to a winning strategy for $\univ(\parobj)$ after a finite prefix of the game.
    This concludes the proof.
\end{proof}

\subsection{Proof of Theorem~\ref{thm:approx-int}}


    \Cref{thm:approx-int}.\ref{it:approxIntFirst} follows directly from \Cref{lem:same-ach} (which says that $\formula(\vect p) \wedge \univ(\parobj)$ is achievable iff $\formula(\vect p)$ is achievable) and \Cref{thm:known_results}.\ref{it:mo_reach} (which asserts that checking if $\formula(\vect p)$ is achievable is in $\PTIME$).
    \Cref{thm:approx-int}.\ref{it:approxIntSecond} follows from the construction in the proof of \Cref{lem:same-ach}.

\section{Sure Parity and Lexicographic Reachability}
\label{sec:lex}

\newcommand{\llex}{<_{lex}}
\newcommand{\leqlex}{\leq_{lex}}
\newcommand{\lexval}{\vect p^*}
\newcommand{\lexvalcomp}{p^*}
\newcommand{\project}[2]{#1^{\pi #2}}

We are now interested in surely satisfying a parity objective while maximizing the probability of reaching $n$ target sets in lexicographic order.
Towards this goal we define the notion of \emph{projection} in \Cref{def:proj}, a concept also used extensively in \Cref{sec:nonstrict}.
We then propose an algorithm using projection and prove it correct.

    Recall that the \emph{lexicographic order} on $[0,1]^n$ is the total order defined as $\vect{x}\llex\vect{y}$ iff there is $k\in [n]$ such that (i) $x_k<y_k$ and (ii) $x_i=y_i$ for all $i\in[k{-}1]$.
    In the following, the order of our target sets $\target_1,\ldots,\target_n$ is relevant:
    For all $i,j \in [n]$,  $F_i$ appears before $F_j$ iff $F_i$ is more important than $F_j$.

One of the difficulties is that when considering the set of achievable points, the lexicographic supremum may not be achievable, i.e., the lexicographic maximum may not exist.
We now formally give the main result of this section.

\begin{restatable}{theorem}{mthlex}
\label{thm:mth}
\label{thm:lex}
    \begin{mdframed}[innertopmargin=+0.5em,innerbottommargin=+0.5em]
Let MDP $\MDP$ be clean w.r.t.\ parity objective $\parobj$ and target sets $\target_1,\ldots,\target_n \subseteq \states$, and let $s\in S$ be a state. Then:
\begin{itemize}
    \item It is decidable if $\lexval = \max_{lex} \{ \vect p \in [0,1]^n \mid  \satisfies{s}{\MDP}\univ(\parobj) \land \bigwedge_{i=1}^n \threshobj{\geq}{p_i}{\event \target_i}\}$ exists by solving $\mathcal{O}(\poly(n))$ many parity games (hence the problem is in \NPinter~\textnormal{\cite{brassard1979note}}).
    \item A witnessing strategy using at most $2|\MDP||\parobj|$ memory can be effectively constructed for the \yes-instances.
\end{itemize}
    \end{mdframed}
\end{restatable}

Our approach considers every target set $F_i$ one by one, following the lexicographic order.
The general idea is to successively remove all transitions that do not achieve the maximal probability to reach $F_i$. 
Thus, after having pruned transitions w.r.t.\ the first $i$ target sets, any strategy that maximizes the probability to reach the set $F = F_1 \cup \ldots \cup F_n$ also maximizes the probabilities of reaching $F_1, \ldots, F_i$ lexicographically~\cite{DBLP:conf/cav/ChatterjeeKWW20}. 
In order to find the maximal probability to achieve a single objective, we adapt the notion of \emph{projection} from~\cite{etessami2007multi,DBLP:conf/atva/ForejtKP12}. 
The main difference is that we keep reachability objectives, instead of converting them into reward objectives, enabling us to use existing results~\cite{icalp2017} on combinations of sure and almost-sure objectives.

We define the MDP $\MDP^{\pi \vect{v}}$, the projection of MDP $\MDP$ on a non-zero vector $\vect{v}\in[0,1]^n$ where we can freely assume $\|\vect v\|_1 = 1$. 
Intuitively, to obtain $\MDP^{\pi \vect{v}}$, we consider a $k$-dimensional face of the Pareto frontier of $\bigwedge_{i=1}^n \Pr_{\geq p_i}(\event F_i)$, maximal in the direction $\vect{v}$. This is thus an intersection with a hyperplane, and defines a face of dimension $k$.
We remove all available actions that are used in none of the strategies achieving this face of dimension $k$, i.e. we remove all non-optimal actions when trying to maximize in the direction $\vect{v}$.
Our purpose is to obtain a new MDP, in which every strategy that almost-surely reaches a final state in $F$ also maximizes the probability to reach these states weighted with the direction $\vect{v}$.
We remark that in this new MDP, the parity condition $\parobj$ may not be surely satisfied  from every state; we will thus need to address this condition later.

\begin{definition}[Projection]
\label{def:proj}
        Let $\MDP = \MDPtuple$ be clean w.r.t.\ $\target_1,\ldots,\target_n \subseteq \states$.
    The \emph{projection} $\project{\MDP}{\vect v}$ of $\MDP$ in direction $\vect v \geq \vect 0$ with $\|\vect v\|_1 = 1$, is defined in two steps:
    (1) Let $\MDP' = (\states', \actions, \tprob')$ be an MDP where
    \begin{itemize}
        \item $\states' = \states \cup \{\projSinkState\}$ where $\projSinkState$ is a fresh sink state, and
        \item $\tprob'$ is defined similar to $\tprob$ with the following modifications (let $\target = \bigcup_{i=1}^n \target_i$):
        \begin{itemize}
            \item For all $s \in \states \setminus \target$, $a \in \actions(s)$, and $s' \in \target$, we set\\
            $
                \tprob'(s,a,s') = \tprob(s,a,s') \cdot \sum_{i : s'\in \target_i} v_i$
                and
               $ \tprob'(s,a,\projSinkState) = 1 - \sum_{s'' \in \states} \tprob'(s,a,s'')
                .
            $
            \item $\tprob'(\projSinkState, a, \projSinkState) = 1$, where $a \in \actions$ is arbitrary.
        \end{itemize}
    \end{itemize}
    (2) For each state $s \in \states'$, let $y_s = \max_{\strat} \prob_s^{\MDPtoMC{\MDP'}{\strat}}(\event \target)$ be the maximum probability to reach $\target$ from $s$ in $\MDP'$.
    The MDP $\project{\MDP}{\vect v}$ is then obtained from $\MDP'$ by removing all actions $a \in \actions(s)$ that do not satisfy $y_s = \sum_{s' \in \states'} \tprob'(s,a,s') \cdot y_{s'}$.
\end{definition}

\begin{example}
    The MDP in \Cref{fig:mdp-pareto-lex} results from projecting the MDP from \Cref{fig:mdp-pareto} on $\vect v = (0,1)$ ($\projSinkState$ is not reachable).
    Only the actions reaching the \tcolBName\ target $F_2$ with maximal probability remain.
\end{example}

\begin{figure}[t]
    \centering
    \adjustbox{max width=\textwidth}{
        \begin{tikzpicture}[on grid, node distance=17.5mm and 20mm]
        \node[player,initial, initial left,initial text=] (s) at (0,0) {$1$} ;
        \node[player,above=of s] (s1) {$1$} ;
        \node[player,below=of s] (s2) {$1$} ;
        \node[player,accepting,right=of s,fill=\tcolA] (r1) {$F_1$} ;
        \node[player,accepting,right=of r1,fill=\tcolB] (r2) {$F_2$} ;
        \node[player,accepting,left=of s,double color fill={\tcolA}{\tcolB}] (r12)  {$F_1\& F_2$} ;
        
        \path[-latex]  (s) edge node[left] { $pair_1,1$} (s1)
        (s) edge[,dashed] node[left] { $pair_2,1$} (s2)
        (s1) edge[bend left=25] node[above right] { $\textcolor{green!50!black}{b},1$} (r2)
        (s2) edge[bend right=25] node[below right, align=left] { $\textcolor{red}{a},\nicefrac 1 3$} (r2)
        (s2) edge[bend left=20] node[below left, align=left] { $\textcolor{red}{a}, \nicefrac 1 6$} (r12)
        (s2) edge [loop below] node[left, pos=0.8] { $\textcolor{red}{a}, \nicefrac 1 2$} (s2)
        (r12) edge [loop left] node[above, pos=0.8] {$*,1$} (r12)
        (r1) edge [loop right] node[above, pos=0.2] { $*,1$} (r2)
        (r2) edge [loop right] node[above, pos=0.2] { $*,1$} (r2)
        ;
        \end{tikzpicture}
        \hspace{2em}
        \begin{tikzpicture}[scale=1.0]
            \begin{axis}[
                clip=false,
                axis lines=center,
                xmin = -0.025, xmax = 1.1,
                ymin = -0.025, ymax = 1.1,
                x=40mm, y=40mm,
                xtick={0,0.2,...,1}, ytick={0,0.2,...,1},
                xlabel={$\prob(\event \tikz[baseline=-0.7ex]{ \node[state,accepting,scale=0.33,fill=\tcolA] {} })$}, xlabel style={right},
                ylabel={$\prob(\event \tikz[baseline=-0.7ex]{ \node[state,accepting,scale=0.33,fill=\tcolB] {} })$}, ylabel style={left},
            ]
            \addplot[color=black,dashed] coordinates { (0.33,0) (0.33,1)};
            \addplot[color=black,dashed] coordinates { (0,1) (0.33,1)};
            \node[label={45:{}},circle,fill,inner sep=2pt] at (axis cs:0,1) {};
            \end{axis}
        \end{tikzpicture}
    }
    \caption{
        The MDP from \Cref{fig:mdp-pareto} projected on $\vect v = (\colorbox{\tcolA}{0},\colorbox{\tcolB}{1})$.
        \emph{Right:} The set of achievable points w.r.t.\ $\univ(\parobj) \land \threshobj{\geq}{x_1}{\event \colorbox{\tcolA}{$\target_1$} } \land \threshobj{\geq}{x_2}{\event \colorbox{\tcolB}{$\target_2$}}$ is $[0,\nicefrac 1 3) \times [0,1) \cup \{(0,1)\}$; the lexicographic maximum for order \colorbox{\tcolB}{$\target_2$}, \colorbox{\tcolA}{$\target_1$} is thus $\vect{p}^* = (\colorbox{\tcolB}{$1$},\colorbox{\tcolA}{$0$})$.
    }
    \label{fig:mdp-pareto-lex}
\vspace{-15pt}
\end{figure}
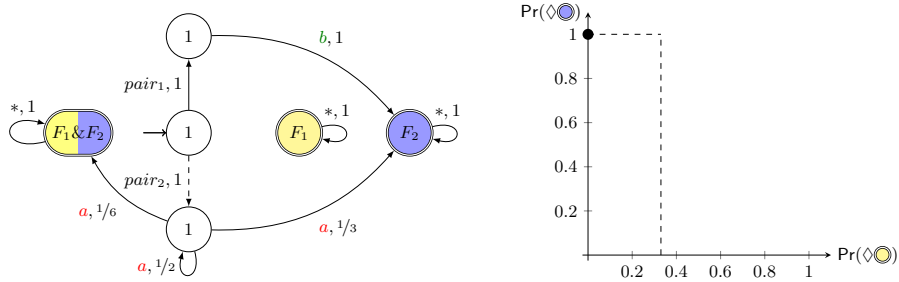

Given $\MDP$, $\target_1,\ldots,\target_n$ and $\vect v$, it is clear from \Cref{def:proj} that we can construct the projection $\project{\MDP}{\vect v}$ in polynomial time.
Note that strategies in $\project{\MDP}{\vect v}$ are still valid in $\MDP$, but the converse is not necessarily the case as projection removes actions.

\subsection{Proof of \Cref{thm:lex}}
\label{app:lex}

We start with the following Lemma.

\begin{lemma}
\label{lem:mdp_proj}
Let MDP $\MDP = \MDPtuple$ be clean w.r.t.\ $\target_1,\ldots,\target_n \subseteq \states$. For every $\vect{v}\in [0,1]^n$, the MDP $\MDP^{\pi \vect{v}}$ is computable.
\end{lemma}

\begin{proof}
To decide which actions to prune in $\MDP$, we consider MDP $\MDP'$ as in \Cref{def:proj}.
We remark that strategies in $\MDP$ and in $\MDP'$ coincide since the new sink state $\projSinkState$ of $\MDP'$ only appears in transitions of $\MDP$ already going to a sink, and every other state and transition is identical in $\MDP$ and $\MDP'$. 
In addition, for all strategies $\sigma$, given the maximal $\vect{p}$ such that $\satisfy{s}{\sigma}{\MDP'} \Pr_{\geq \vect{p}}(\event F)$, and among the $\vect{x}\in[0,1]$ such that $\satisfy{s}{\sigma}{\MDP}\limwedgeone{i \in [n]} \Pr_{\geq x_i}(\event F_i)$, any $\vect{x}$ that maximizes  $\vect{x}\mapsto \vect{x}\cdot\vect{v}$ satisfies $\vect{p}=\vect{x}\cdot\vect{v}$. 
Indeed, any transition of $\MDP$ to some $s'\in F_i$ with probability $x$, goes to $s'$ with probability $x\cdot v_i$ in $\MDP'$. 
Thus $\satisfy{s}{\sigma}{\MDP'} \Pr_{\geq p^*}(\event F)$ where $p^*$ is the maximum achievable value if and only if both $\vect{x}\mapsto \vect{x}\cdot\vect{v}$ is maximal and $\satisfy{s}{\sigma}{\MDP}\limwedgeone{i \in [n]} \Pr_{\geq x_i}(\event F_i)$ for some $\vect{x}\in[0,1]$.

For all states, we can compute all the actions that could be taken by some strategy $\sigma$ such that $\satisfy{s}{\sigma}{\MDP'} \Pr_{\geq x^*}(\event F)$. 
These actions coincide with the actions taken by a strategy $\sigma$ such that $\vect{x}\mapsto \vect{x}\cdot\vect{v}$ is maximal, and $\satisfy{s}{\sigma}{\MDP}\limwedgeone{i \in [n]} \Pr_{\geq x_i}(\event F_i)$ for some $\vect{x}\in[0,1]$. 
By removing all the other actions from every state, we obtain $\MDP^{\pi \vect{v}}$.
\end{proof}

It remains to show that any strategy that almost-surely reaches a target in $\MDP^{\pi \vect{v}}$ is also maximizing $\vect{x}\mapsto \vect{x}\cdot\vect{v}$ in the original MDP $\MDP$. Since these MDPs share strategies maximizing $\vect{x}\mapsto \vect{x}\cdot\vect{v}$ by construction, this is straightforward. We give the formal proof next.

\begin{restatable}{lemma}{lemproj}(Key property of projection)
\label{lem:mdp_proj_c}
Let $\MDP= \MDPtuple$ be clean w.r.t.\ $\target_1,\ldots,\target_n,\projSinkState \subseteq \states$, and let $\vect v \geq \vect 0$, $\|\vect v\|_1 = 1$.
Then for all strategies $\sigma$ of $\project{\MDP}{\vect v}$, the following holds: 
$\satisfy{s}{\sigma}{\MDP^{\pi \vect{v}}} \threshobj{=}{1}{\event \target}$ iff there exists $\vect{x}\in [0,1]^n$ such that 
(i) $\vect{x} \cdot \vect{v}$ is maximal among the achievable $\vect{x}$, and
(ii) $\satisfy{s}{\sigma}{\MDP}\bigwedge_{i = 1}^n \threshobj{\geq}{x_i}{\event \target_i}$.
\end{restatable}

\begin{proof}
    We first show the left to right direction. 
    Let $\sigma$ such that  $\satisfy{s}{\sigma}{\MDP^{\pi \vect{v}}}\Pr_{= 1}(\event F)$. 
    Let $\MDP'$ be as in Definition~\ref{def:proj}, hence $\satisfy{s}{\sigma}{\MDP'}\Pr_{= 1}(\event (F \cup \projSinkState))$. There exists some optimal value $p^*$ and some memoryless strategy $\sigma'$ such that $\satisfy{s}{\sigma'}{\MDP}\Pr_{\geq p^*}(\event F)$. Since only actions maximizing the probability to reach $F$ remain in $\MDP^{\pi \vect{v}}$, we have $p^* =\prob_{\MDP'[\sigma]}(\event F)/\prob_{\MDP'[\sigma]}(\event (F \cup \projSinkState)) = \prob_{\MDP'[\sigma]}(\event F)$. 
    Thus $\satisfy{s}{\sigma}{\MDP}\Pr_{\geq p^*}(\event F)$. 
    As a consequence, there exists some $\vect{x}$ such that $\vect{x}\cdot \vect{v}$ is maximal among the achievable $\vect{x}$, and $\satisfy{s}{\sigma}{\MDP^{\pi \vect v}}\limwedgeone{i \in [n]} \Pr_{\geq x_i}(\event F_i)$.
    
    For the right to left direction, let $\sigma$ such that $\satisfy{s}{\sigma}{\MDP}\limwedgeone{i \in [n]} \Pr_{\geq x_i}(\event F_i)$ with $\vect{x}.\vect{v}$ maximal among the achievable $\vect{x}$. Since $\MDP$ is clean there is some $\sigma'$ such that $\satisfy{s}{\sigma'}{\MDP}\Pr_{=1}(\event F)$. 
    For every vector $\vect{x}'$ such that  $\satisfy{s}{\sigma'}{\MDP}\limwedgeone{i \in [n]} \Pr_{\geq x'_i}(\event F_i)$ we have $\vect{x}'\cdot \vect{v}\leq \vect{x}\cdot \vect{v}$. 
    Thus by considering $\MDP'$ we have for every $p\in [0,1]$, if $\satisfy{s}{\sigma'}{\MDP'}\Pr_{\geq p}(\event F)$ then $\satisfy{s}{\sigma}{\MDP'}\Pr_{\geq p}(\event F)$. 
    Since $\MDP^{\pi \vect{v}}$ is obtained by removing states of $\MDP'$ that neither $\sigma$ nor $\sigma'$ visits, and since $\satisfy{s}{\sigma'}{\MDP^{\pi \vect{v}}}\Pr_{=1}(\event F)$ we have $\satisfy{s}{\sigma}{\MDP^{\pi \vect{v}}}\Pr_{=1}(\event F)$.
\end{proof}

We can now give an algorithm computing a strategy achieving a sure parity condition while maximizing lexicographically ordered reachability conditions. Remark that in this section we only need projection on unit vectors, the general construction is only used in \Cref{sec:nonstrict}. 
With Algorithm~\ref{proc:mth}, we can then prove Theorem~\ref{thm:mth}.

\begin{algorithm}[t]
    \caption{Sure parity and lexicographic reachability}
    \label{alg:lex}
    \label{proc:mth}
    \textbf{Input:} MDP $\MDP$ -- clean w.r.t.\ parity objective $\parobj$ and $\target_1,\ldots,\target_n \subseteq \states$, a state $s \in \states$\\
    \textbf{Output:} If $\lexval = \max_{lex} \{ \vect p \in [0,1]^n \mid  \satisfies{s}{\MDP}\univ(\parobj) \land \bigwedge_{i=1}^n \threshobj{\geq}{p_i}{\event \target_i}\}$ exists, then the output is a witness strategy $\strat$, otherwise the output is \texttt{false}.
    \begin{algorithmic}[1]
        \State \label{proc-mth-1}$\MDP_0 \gets \MDP$ 
        \State $\target \gets \bigcup_{i=1}^n \target_i$        \For {$i$ from $1$ to $n$} 
        \State \label{proc-mth-2}Compute $\MDP_{i-1}^{\pi \vect{e}_i}$  \Comment{See \Cref{def:proj}.}
        \State \label{proc-mth-3} $\MDP_i \gets$ result of pruning all states not satisfying $\univ(\parobj) \land \threshobj{=}{1}{\event \target}$ in $\MDP_{i-1}^{\pi \vect{e}_i}$.
        \EndFor
        \If{$s$ is not a state of $\MDP_n$}
        \Return \texttt{false}
        \Else{}
        \Return \label{proc-mth-4}$\sigma$ such that $\satisfy{s}{\sigma}{\MDP_n}\univ(\parobj)\wedge  \Pr_{= 1}(\event \target)$ \Comment{By \Cref{thm:mth}.}
        \EndIf
    \end{algorithmic}
\end{algorithm}

\begin{example}
\label{ex:algo1}
    We follow Algorithm~\ref{proc:mth} and project the MDP of Figure~\ref{fig:mdp-pareto} on $e_1=(1,0)$. After this projection, there only remains the uppermost state with action $\textcolor{red}{a}$, associated to point $(1,1/3)$. This projected MDP does not satisfy $\univ(\parobj)$ and we conclude that there is no lexicographically optimal point for the order $(e_1,e_2)$.    
    
    We now show what happens if we start projecting on $e_2=(0,1)$. We obtain the MDP in Figure~\ref{fig:mdp-pareto-lex}, with both full and dashed transitions: only the points that reach the blue state $F_2$ with maximal probability remain. Since the lowermost state does not satisfy the parity condition, it refutes $\univ(\parobj) \wedge \Pr_{= 1}(\event F)$, and we also prune it, only keeping the full transitions. We thus only keep point $(0,1)$ that can be achieved while satisfying parity condition $\parobj$, and is thus optimal for this order.
\end{example}


\begin{proof}[of \Cref{thm:lex}]
We follow Algorithm~\ref{proc:mth}. 
During iteration $i$, in line~\ref{proc-mth-2} we project the current MDP $\MDP_{i-1}$ on the next target set $F_i$, obtaining $\MDP_{i-1}^{\pi \vect{e}_i}$. 
By Lemma~\ref{lem:mdp_proj_c}, a strategy satisfies $\Pr_{= 1}(\event F)$ in $\MDP_{i-1}^{\pi \vect{e}_i}$ iff it maximizes $p$ such that $\Pr_{\geq p}(\event F_i)$ in $\MDP_{i-1}$.
We then prune all states that do not satisfy the conjunction with the parity condition $\univ(\parobj) \wedge \Pr_{= 1}(\event F)$. 

If the algorithm outputs a strategy $\sigma$, it follows $\satisfy{s}{\sigma}{\MDP_n}\univ(\parobj)\wedge \Pr_{= 1}(\event F)$ and is maximal for $<_{lex}$ on $\MDP$. 
Let $\vect{x}$ be the largest vector such that $\satisfy{s}{\sigma}{\MDP}\limwedgeone{i \in [n]} \Pr_{\geq x_i}(\event F_i)$. 
We now take any $\sigma'$, and the associated largest $\vect{x}'$ and  $m\in[n]$ such that for all  $i\in{[m]_0}$ we have $x_i = x'_i$, but not for $m+1$, and such that $\satisfy{s}{\sigma'}{\MDP}\limwedgeone{i \in [n]} \Pr_{\geq x'_i}(\event F_i)$ (if such a $\sigma'$ does not exist, all strategies yield the same value and so $\sigma$ is optimal). 
Thus both $\sigma$ and $\sigma'$ are well defined on $\MDP_i$ and since $\sigma$ is a strategy of the projected MDP $\MDP_m^{\pi \vect{e}_{m+1}}$, by \Cref{lem:mdp_proj_c} we have that $x'_{i+1}\leq x_{i+1}$, proving that $\sigma$ is maximal for $<_{lex}$ on $\MDP$.

If the output is \texttt{false}, since only Step~\ref{proc-mth-3} can prune states, this means that there was some strategy lexicographically maximizing the $k$ first components and satisfying $\univ(\parobj)$, but of all these strategies, those also maximizing the $(k+1)$th component did not satisfy $\univ(\parobj)$, hence there was no lexicographically optimal strategy satisfying $\univ(\parobj)$.

We show that Algorithm~\ref{proc:mth} solves a polynomial number of parity games of size polynomial in $|\Gamma|$. 
Indeed, $|\MDP_{i-1}^{\pi \vect{v}_i}|$ only has one more state than $\Gamma$ and at most $|\Gamma|$ additional transitions (that may lead from states originally in $\Gamma$ to the new state $\projSinkState$). 
We then remove all states that do not satisfy $\univ(\parobj) \wedge \Pr_{= 1}(\event F)$, which can be done by solving a polynomial number of parity games~\cite{icalp2017}. 
We do this step $n$ times. 
Since there are $n$ targets in $\Gamma$ we have $n\leq |\Gamma|$. 
Thus we end up solving a number of parity games polynomial in $\Gamma$.

 If the output is \texttt{true}, using~\cite{icalp2017} we can generate optimal finite-memory strategies. We give a precise bound on the amount of memory needed: such a finite-memory strategy is obtained through a reduction to a parity-B\"uchi game, which, as shown in the proof of Lemma~3 of~\cite{almagor2016minimizing}, can be reduced to a mean-payoff parity game with a strict threshold of $0$. It has been shown in Theorem~4 of~\cite{DBLP:journals/tcs/ChatterjeeD12} that if there exists a strategy for strict threshold $0$, there also exists a strategy for threshold $\frac{1}{|\MDP|}$ which implies there exists a winning strategy in the energy parity game with energy value $1-\frac{1}{|\MDP|(|\MDP|+1)}$ for B\"uchi states and energy value $-\frac{1}{|\MDP|(|\MDP|+1)}$ for non-B\"uchi states. By Lemma~5 of~\cite{DBLP:journals/tcs/ChatterjeeD12} such energy parity games require winning strategies with memory at most $2|\MDP||\parobj|$. These strategies need to visit B\"uchi states infinitely often to satisfy their energy condition, and so are winning the parity-B\"uchi game initially defined. Winning strategies for energy parity games can be computed in $\NPinter$. A quasi-polynomial algorithm exists for mean-payoff parity games~\cite{DBLP:conf/lics/DaviaudJL18} (which is indeed quasi-polynomial since there are only two possible payoffs for every state, which is bounded: payoffs are either $0$ or $1$), but it is unclear if it can be used to find such polynomial-memory strategies.
\end{proof}


%

\section{Sure Parity and $n$ non-Strict Reachability Thresholds}
\label{sec:multi}
\label{sec:nonstrict}

Finally, we consider the case of one sure parity condition and multiple \emph{non-strict} threshold reachability objectives, i.e., formulas like $\univ(\parobj)\bigwedge_{i = 1}^n \threshobj{\geq}{p_i}{\event \target_i}$.
We do not impose a lexicographic ordering on the target sets. Our main result is:

\begin{restatable}{theorem}{thstrict}
\label{thm:mult-th}
    \begin{mdframed}[innertopmargin=+0.5em,innerbottommargin=+0.5em]
    Let MDP $\MDP$ be clean w.r.t.\  parity objective $\parobj$ and target sets $\target_1,\ldots,\target_n \subseteq \states$.
    Further, let $\vect{p} \in [0,1]^n$, and let $s \in S$ be a state.
    Then:
        \begin{enumerate}
        \item The decision problem $\exists\strat \colon \satisfy{s}{\strat}{\MDP} \univ(\parobj) \wedge \bigwedge_{i = 1}^n \threshobj{\geq}{p_i}{\event \target_i}$ is in $\EXPTIME$.
        \item A witness strategy $\strat$ using at most $2^{\poly(|\MDP| + nD)}$ memory, where $D$ is the bit-complexity of the rational numbers in $\vect p$, can be effectively constructed for the \yes-instances.
    \end{enumerate}
    \end{mdframed}
\end{restatable}

We solved the case where $\vect{p}$ is strictly inside the Pareto frontier in Section~\ref{sec:thres-r}.  It remains to show how to achieve $\univ(\parobj) \bigwedge_{i = 1}^n \threshobj{\geq}{p_i}{\event \target_i}$ when $\vect{p}$ is exactly on the frontier. We first consider the case where $\vect{p}$ is a vertex of the Pareto frontier, that we will then use as a base case for an arbitrary point $\vect{p}$ of the frontier.  We sketch the proof in the remainder of the section.

Since the Pareto frontier does not depend on the sure objective (cf.~\Cref{lem:same-ach}), to determine whether $\vect{p}$ is exactly on the Pareto frontier, it suffices to check if $\satisfies{s}{\MDP} \bigwedge_{i = 1}^n \threshobj{\geq}{p_i}{\event \target_i}$ and $\nsatisfies{s}{\MDP}\bigwedge_{i = 1}^n \threshobj{>}{p_i}{\event \target_i}$. The first formula checks if $\vect{p}$ is achievable, the second checks whether it is on the boundary.
Hence the main difficulty is to decide whether adding a sure parity condition keeps the achievability of a point.
 We illustrate this in the following example.

\begin{example}
In the MDP of Figure~\ref{fig:mdp-pareto}, playing $pair_1$ then $\textcolor{red}{a}$ forever gives a total probability of $1$ of reaching $F_1$ and probability $\nicefrac 1 3$ of reaching $F_2$. This strategy does not surely satisfy the parity condition though, since there exists a path that visits the uppermost state, labelled $1$, forever.
Playing $pair_1$ and then $\textcolor{green!50!black}{b}$ once surely reaches $F_2$. 
It is thus possible to satisfy $\univ(\parobj) \wedge \threshobj{=}{1}{\event \target_2}$, but not $\univ(\parobj)\wedge\threshobj{=}{1}{\event \target_1} \wedge \threshobj{\geq}{\nicefrac 1 3}{\event \target_2}$.
Still, for every $\varepsilon > 0$, we can achieve $\univ(\parobj)\wedge \threshobj{\geq}{1- \varepsilon}{\event \target_1} \wedge \threshobj{\geq}{\nicefrac 1 3 - \varepsilon}{\event \target_2}$ by Theorem~\ref{thm:approx-int}, using a finite-memory strategy.
\end{example}

\subsection{Vertex of the Pareto frontier}
\label{sec:vertex}
 
We first consider the easier case where we want to achieve a point which is a vertex of the Pareto frontier. 
We assume $\vect{p}$ to be a vertex of the Pareto frontier.  Our proof relies on projection (\Cref{def:proj}).
Indeed, since  $\vect{p}$ is a vertex, there exists some vector $\vect{v}$ such that  $\vect{p}$ is the \emph{unique} point of the Pareto frontier maximizing $\vect{p}\cdot \vect{v}$. 
We obtain the following lemma.

\begin{restatable}{lemma}{lemvertex}
\label{lem:vert}
Suppose that the MDP $\MDP$ is clean w.r.t.\ parity objective $\parobj$ and target sets $\target_1,\ldots,\target_n \subseteq \states$.
Further, let $\vect{p} \in [0,1]^n$, and let $s \in S$ be a state.
If $\vect{p}$ is a vertex of the Pareto frontier of $\bigwedge_{i = 1}^n \threshobj{\geq}{p_i}{\event \target_i}$ from $s$, then we can decide if $\satisfies{s}{\MDP}\univ(\parobj)\wedge\bigwedge_{i = 1}^n \threshobj{\geq}{p_i}{\event \target_i}$ and if so give a finite-memory strategy.
\end{restatable}

\begin{proof}
Since $\vect{p}$ is a vertex of the Pareto frontier $P$, which is a convex polytope, there exists a half-space $H=\{\vect{x}\mid \vect{w}\cdot \vect{x} \leq c\} \supseteq P$ where $c\in[0,1]$, $\vect{w} \geq \vect 0$, $\|\vect w\|_1 = 1$, and $\vect{p}$ is the only element of $P$ satisfying the constraint defining $H$ with \emph{equality}, that is $\vect w \cdot  \vect p = c$. 
By \Cref{lem:mdp_proj_c}, $\satisfy{s}{\sigma}{\MDP^{\pi \vect{w}}}\threshobj{\geq}{1}{\event\target}$ iff $\satisfy{s}{\sigma}{\MDP}\bigwedge_{i = 1}^n \threshobj{\geq}{x_i}{\event \target_i}$ for all strategies $\strat$ in $\MDP^{\pi \vect{w}}$, and thus also
\[
    \satisfy{s}{\sigma}{\MDP^{\pi \vect{w}}} \univ(\parobj) \wedge \threshobj{\geq}{1}{\event\target}
    \quad\text{iff}\quad
    \satisfy{s}{\sigma}{\MDP} \univ(\parobj) \wedge \bigwedge_{i = 1}^n \threshobj{\geq}{x_i}{\event \target_i}
    ~.
\]
Thus, using \Cref{thm:known_results}.\ref{thm:reach-p}, we can decide whether $\satisfies{s}{\MDP^{\pi \vect{w}}} \univ(\parobj) \wedge \threshobj{\geq}{1}{\event\target}$ holds and if so, obtain a finite-memory strategy $\sigma$.
\end{proof}

\subsection{Arbitrary Point of the Pareto Frontier}

We now consider any arbitrary point $\vect{p}$ of the Pareto frontier.
Since $\vect{p}$ may be contained in a $k$-dimensional face of the frontier (with $k > 0$; $k=0$ means that $\vect p$ is a vertex, see \Cref{sec:vertex}), projecting on this face will not be sufficient to obtain $\vect{p}$.
Similar to \Cref{alg:lex}, we iterate projections and state removal, thereby reducing the dimension of the Pareto frontier until either reducing $\vect{p}$ to a vertex, or entering a situation where we cannot project anymore.
We remark that the latter only happens in specific cases. To properly define these cases, given a Pareto frontier $P$, we consider the smallest vector space containing $P$. We show that we cannot project any more when $\vect{p}$ is an interior point of $P$ for this subspace, denoted $\vect{p}\in\interior_{\generator(P)}(P)$. We obtain \Cref{proc:mthc}.

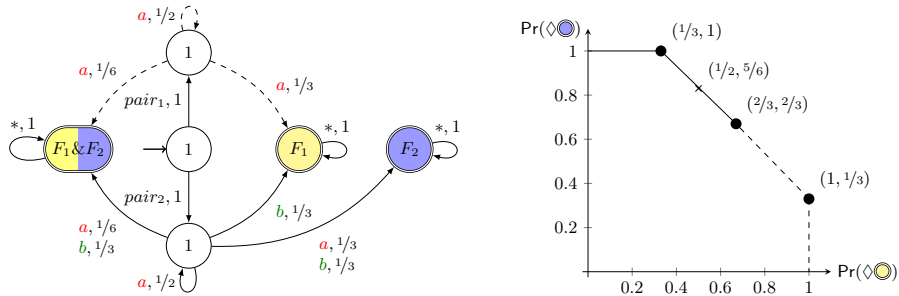
\begin{figure}[t]\scalebox{1}
    \centering
    \adjustbox{max width=\textwidth}{
        \begin{tikzpicture}[on grid, node distance=17.5mm and 20mm]
        \node[player,initial, initial left,initial text=] (s) at (0,0) {$1$} ;
        \node[player,above=of s] (s1) {$1$} ;
        \node[player,below=of s] (s2) {$1$} ;
        \node[player,accepting,right=of s,fill=\tcolA] (r1) {$F_1$} ;
        \node[player,accepting,right=of r1,fill=\tcolB] (r2) {$F_2$} ;
        \node[player,accepting,left=of s,double color fill={\tcolA}{\tcolB}] (r12)  {$F_1\& F_2$} ;
        
        \path[-latex]  (s) edge node[left] { $pair_1,1$} (s1)
        (s) edge node[left] { $pair_2,1$} (s2)
        
        (s1) edge[bend left=20, dashed] node[above right, pos=0.7] { $\textcolor{red}{a}, \nicefrac 1 3$} (r1)
        (s1) edge[bend right=20, dashed] node[above left] { $\textcolor{red}{a}, \nicefrac  1 6$} (r12)
                (s2) edge[bend right=25] node[below right, align=left] { $\textcolor{red}{a}, \nicefrac 1 3$ \\ $\textcolor{green!50!black}{b}, \nicefrac 1 3$} (r2)
        (s2) edge[bend left=20] node[below left, align=left] { $\textcolor{red}{a}, \nicefrac 1 6$ \\ $\textcolor{green!50!black}{b}, \nicefrac 1 3$} (r12)
        (s2) edge[bend right=20] node[below right, pos=0.7] { $\textcolor{green!50!black}{b}, \nicefrac 1 3$} (r1)
        (s1) edge [loop above, dashed] node[left, pos = 0.2] { $\textcolor{red}{a}, \nicefrac 1 2$} (s1)
        (s2) edge [loop below] node[left, pos=0.8] { $\textcolor{red}{a}, \nicefrac 1 2$} (s2)
        (r12) edge [loop left] node[above, pos=0.8] {$*,1$} (r12)
        (r1) edge [loop right] node[above, pos=0.2] { $*,1$} (r2)
        (r2) edge [loop right] node[above, pos=0.2] { $*,1$} (r2)
        ;
        \end{tikzpicture}
        \hspace{2em}
        \begin{tikzpicture}[scale=1.0]
            \begin{axis}[
                clip=false,
                axis lines=center,
                xmin = -0.025, xmax = 1.1,
                ymin = -0.025, ymax = 1.1,
                x=40mm, y=40mm,
                xtick={0,0.2,...,1}, ytick={0,0.2,...,1},
                xlabel={$\prob(\event \tikz[baseline=-0.7ex]{ \node[state,accepting,scale=0.33,fill=\tcolA] {} })$}, xlabel style={right},
                ylabel={$\prob(\event \tikz[baseline=-0.7ex]{ \node[state,accepting,scale=0.33,fill=\tcolB] {} })$}, ylabel style={left},
            ]
    \addplot[color=black] coordinates { ((0.67,0.67) (0.33,1)};
    \addplot[color=black] coordinates { ((0,1) (0.33,1)};
    \addplot[color=black, dashed] coordinates { (1,0.33) (0.67,0.67)};
    \addplot[color=black, dashed] coordinates { (1,0.33) (1,0)};
                \node[label={45:{$(1,\nicefrac 1 3)$}},circle,fill,inner sep=2pt] at (axis cs:1,0.33) {};
    \node[label={45:{$(\nicefrac 2 3,\nicefrac 2 3)$}},circle,fill,inner sep=2pt] at (axis cs:0.67,0.67) {};
    \node[label={45:{$(\nicefrac 1 3,1)$}},circle,fill,inner sep=2pt] at (axis cs:0.33,1) {};
    \node[label={45:{$(\nicefrac 1 2,\nicefrac 5 6)$}},cross=2pt] at (axis cs:0.5,0.83) {};
            \end{axis}
        \end{tikzpicture}
    }
\caption{\label{fig:mdp-pareto-gen} The MDP of Figure~\ref{fig:mdp-pareto} after projection on $\vect v = (\colorbox{\tcolA}{1},\colorbox{\tcolB}{1})$.
}
\vspace{-5pt}
\end{figure}

\begin{example}
\label{ex:proj}
After projecting the MDP of \Cref{fig:mdp-pareto} on $\vect v = (1,1)$, we obtain the MDP in \Cref{fig:mdp-pareto-gen} (including both the dashed and the solid transitions).
Since from the uppermost state, no strategy surely satisfies the parity condition, $\univ(\parobj) \wedge \threshobj{=}{1}{\event\target}$ does not hold.
We thus prune the dashed transitions, and the Pareto frontier is now restricted between $(\nicefrac 1 3,1)$ and $(\nicefrac 2 3,\nicefrac 2 3)$, see Figure~\ref{fig:mdp-pareto-gen} (right).
To achieve e.g.\ $\vect p = (\nicefrac 1 2,\nicefrac 5 6)$, which is strictly inside this line segment, it suffices to play $pair_2$, then $\textcolor{red}{a}$ once, and then finally $\textcolor{green!50!black}{b}$ once.
\end{example}

To obtain our result, we prove the following. 
After projecting on a given vector, and removing any state refuting $\univ(\parobj) \wedge \threshobj{=}{1}{\event\target}$ we obtain some polytope $P$; any point of $P$ that is a topologically interior point in the smallest vector space containing $P$ is achievable. Formally:

\begin{restatable}{lemma}{lemarbi}
    \label{lem:int-par}
    Let the MDP $\MDP$ be clean w.r.t.\ parity objective $\parobj$ and target sets $\target_1,\ldots,\target_n \subseteq \states$.
    Further, let $\vect{v} \in [0,1]^n$, and $s \in S$. Let $\MDP_{\parobj}$ be obtained by taking the MDP $\MDP^{\pi \vect{v}}$ and pruning all states that refute $\univ(\parobj) \wedge \threshobj{=}{1}{\event\target}$.
    Let $B\subseteq[0,1]^n$ be the set of $\leq$-maximal points of the Pareto frontier of $\MDP_{\parobj}$ from $s$.
    Then:
    For every $\vect{x}\in\interior_{\generator(B)}(B)$, we have $\satisfies{s}{\MDP_{\parobj}}\univ(\parobj)\wedge \bigwedge_{i = 1}^n \threshobj{\geq}{x_i}{\event \target_i}$, and we can compute a strategy that achieves this.
\end{restatable}

\begin{figure}[t]
\centering
\scalebox{0.8}{
 \begin{tikzpicture}
    \begin{axis}[ axis equal,   
    xmin = 0, xmax = 1.2,
    ymin = 0, ymax = 1.2,
    xlabel={Probability of $F_1$},
    ylabel={Probability of $F_2$}
    ]
                \node[label={45:{$x$}},circle,fill,inner sep=2pt] at (axis cs:0.5,0.5) {};
    \node[label={180:{$y^1$}},circle,fill,inner sep=2pt] at (axis cs:0.1,0.1) {};
    \node[label={90:{$y^2$}},circle,fill,inner sep=2pt] at (axis cs:0.5,0.9) {};
    \node[label={180:{$y^3$}},circle,fill,inner sep=2pt] at (axis cs:0.9,0.5) {};
    \draw[dashed] (axis cs: 0.1, 0.1) circle [radius=0.1];
    \draw[dashed] (axis cs: 0.5, 0.9) circle [radius=0.1];
    \draw[dashed] (axis cs: 0.9, 0.5) circle [radius=0.1];
    \node[label={45:{$z^1$}},circle,fill,inner sep=2pt] at (axis cs:0.17,0.17) {};
    \node[label={45:{$z^2$}},circle,fill,inner sep=2pt] at (axis cs:0.5,0.8) {};
    \node[label={45:{$z^3$}},circle,fill,inner sep=2pt] at (axis cs:1,0.5) {};
    \addplot[color=black, dotted] coordinates { (0.17,0.17) (0.5,0.8)};
    \addplot[color=black, dotted] coordinates { (0.5,0.8) (1,0.5)};
    \addplot[color=black, dotted] coordinates { (1,0.5) (0.17,0.17)};
    \end{axis}
\end{tikzpicture}
}
\vspace{-10pt}
\caption{\label{fig:approx} To obtain $x$ on the two-dimension plane, we first take three points $y^1,y^2,y^3$, then find $\varepsilon$ small enough for $x$ to be within the convex hull of any $\varepsilon$-approximation $z^1,z^2,z^3$ of $y^1,y^2,y^3$.}
\end{figure}
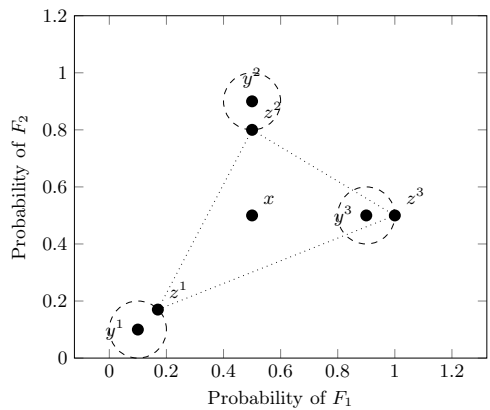

The proof of this lemma is quite involved. Figure~\ref{fig:approx} provides some intuition on the proof. If  $\vect{x}$ is inside an $m$-dimensional surface $B$, we can find $m+1$ elements of $B$ such that $\vect{x}$ is within their convex hull. These are $y^1,y^2,y^3$ in Figure~\ref{fig:approx}. For every $y^j$, we can find a strategy satisfying $\limwedgeone{i \in [n]} \Pr_{\geq y^j_i}(\event F_i)$. By playing such a strategy for sufficiently many steps, then switching to a strategy satisfying $\univ(\parobj) \wedge \Pr_{= 1}(\event F)$, we can $\varepsilon$-approximate the $y^j$ while staying inside $B$. Points achieved by such approximation are denoted $z^1,z^2,z^3$ in Figure~\ref{fig:approx}. It then remains to show that if $\varepsilon$ is small enough, $\vect{x}$ is within the convex hull of $z^1,z^2,z^3$ and thus can be achieved.

\begin{proof}[of \Cref{lem:int-par}]
    We first show that it is sufficient to $\varepsilon$-approximate any point of $B$. Let $m$ be the number of dimensions of $B$.
    Given $\vect{x}\in B$, we have $\vect{x}\in\interior_{\generator(B)}(B)$ if and only if there exist $m+1$ points of $B$, denoted $\vect{y}^1,\cdots,\vect{y}^{m+1}\in B$ such that $\vect{x}$ is strictly inside the convex hull of $\vect{y}^1,\cdots,\vect{y}^{m+1}$. There also exists $\varepsilon>0$ such that for all $\vect{z}^1,\cdots,\vect{z}^{m+1}\in B$, if we have for all $j\in[m+1]$ that $|\vect{y}^j-\vect{z}^j|\leq \varepsilon$ (that is for all $i\in [n]$ we have $|y^j_i-z^j_i|\leq \varepsilon$), then $\vect{x}$ is inside the convex hull of $\vect{z}^1,\cdots,\vect{z}^{m+1}$. If for every $\vect{z}^j$ we can find a strategy satisfying $\satisfies{s}{\MDP_{\parobj}}\limwedgeone{i \in [n]}\univ(\parobj) \wedge \Pr_{\geq z^j_i}(\event F_i)$, then by using a convex combination of such strategies we can get a strategy satisfying $\satisfies{s}{\MDP_{\parobj}}\limwedgeone{i \in [n]} \univ(\parobj)\wedge \Pr_{\geq x_i}(\event F_i)$. Given $\vect{x}$, we can find such $\vect{y}^1,\cdots,\vect{y}^{m+1}$ and such $\varepsilon$. There remains to find for every $\vect{y}^j$ a $\vect{z}^j$ such that $|\vect{y}^j-\vect{z}^j|\leq\varepsilon$.

    We now show that for any $\vect{y}\in B$, for all $\varepsilon> 0$, there exists $\vect{z}\in B$ such that for all $i\in n$ we have $|y_i-z_i|\leq \varepsilon$ and a strategy $\sigma$ such that  $\satisfies{s}{\MDP_{\parobj}}\limwedgeone{i \in [n]}\univ(\parobj) \wedge \Pr_{\geq z_i}(\event F_i)$. 
    The result follows since any interior point $\vect{x}\in\interior_{\generator(B)}(B)$ can be obtained as a convex combination using strict inequalities of such vertices, and thus as a convex combination of some approximations of these vertices. 

    We fix $\vect{y}\in B$ and $\varepsilon >0$. 
    Let $\sigma_y$ be a memoryless strategy such that $\satisfies{s}{\MDP_{\parobj}}\limwedgeone{i \in [n]} \Pr_{\geq y_i}(\event F_i)$. 
    Using Theorem~\ref{thm:opti_reach}, we can get $k\in \mathbb{N}_0$ such that after playing $\sigma_y$ for $k$ rounds, we have probability at least $y_i-\varepsilon$ to have reached $F_i$. 
    There exists some strategy $\sigma_{\parobj}$ such that for all states $s\in S$ we have $\satisfy{s}{\sigma_{\parobj}}{\MDP_{\parobj}} \univ(\parobj)\wedge \Pr_{=1}(\event F)$. 
    We now define strategy $\sigma$ as follows: play $\sigma_y$ for $k$ steps, then from the current state $s$ play $\sigma_{\parobj}$. Strategy $\sigma$ satisfies $\univ(\parobj)$ since it suffices $\sigma_{\parobj}$ does. Now, since $\satisfy{s}{\sigma}{\MDP_{\parobj}}\Pr_{=1}(\event F)$, by Lemma~\ref{lem:mdp_proj_c} there exists some $\vect{z}\in B$ such that $\satisfy{s}{\sigma}{\MDP_{\parobj}}\limwedgeone{i \in [n]} \Pr_{\geq z_i}(\event F_i)$. Since $\satisfy{s}{\sigma}{\MDP_{\parobj}}\limwedgeone{i \in [n]} \Pr_{\geq y_i-\varepsilon}(\event F_i)$, we have $z_i\geq y_i-\varepsilon$ and thus $|z_i-y_i|\leq \varepsilon$, concluding the proof.
\end{proof}

We can now give our result stating that we can verify the achievability of an arbitrary point of the Pareto frontier.

\begin{lemma}
    \label{lem:mult-th}
    Suppose that the MDP $\MDP$ is clean w.r.t.\  parity objective $\parobj$ and target sets $\target_1,\ldots,\target_n \subseteq \states$.
    Further, let $s \in S$ be a state, and $\\{\vect{x} = (x_i)_{i\in [n]}\in [0,1]^n}$ on the border of the Pareto frontier of $\limwedgeone{i \in [n]} \Pr_{\geq x_i}(\event F_i)$
    \begin{itemize}
        \item Checking whether $\satisfies{s}{\MDP}\univ(\parobj)\wedge\limwedgeone{i \in [n]} \Pr_{\geq x_i}(\event F_i)$ is decidable; 
        \item if \yes\ a witnessing strategy with at most $2|\MDP||\parobj|$ memory can be effectively computed.
    \end{itemize}
\end{lemma}
\begin{proof}
The proof of this lemma is similar to the proof of \Cref{thm:mth}. 
We show that \Cref{proc:mthc} answers \Cref{lem:mult-th}. The main differences with \Cref{proc:mth} is that we must find the vector on which we project, and at the end of the loop of \Cref{proc:mthc} we have to split between the cases where $\vect{x}$ is a vertex or not. 

During iteration $i$, we first find a vector orthogonal to the face of $P_i$ that $\vect{x}$ belongs to. To do so, we may need to fully compute the Pareto frontier of $\MDP_{i-1}$.
In line~\ref{proc-mthc-4} we project the current MDP $\MDP_{i-1}$ on vector $\vect{v}_i$, obtaining $\MDP_{i-1}^{\pi \vect{v}_i}$. 
By Lemma~\ref{lem:mdp_proj_c}, a strategy satisfies $\Pr_{= 1}(\event F)$ in $\MDP_{i-1}^{\pi \vect{e}_i}$ iff it maximizes $p$ such that $\Pr_{\geq p}(\event F_i)$ in $\MDP_{i-1}$.
We then prune all states that do not satisfy the conjunction with the parity condition $\univ(\parobj) \wedge \Pr_{= 1}(\event F)$. 

If $\vect{x}$ is an interior point, it follows by Lemma~\ref{lem:int-par} that we can find a strategy $\sigma$ such that $\satisfy{s}{\sigma}{\MDP_i}\univ(\parobj)\wedge\bigwedge_{i = 1}^n \threshobj{\geq}{x_i}{\event \target_i}$. Otherwise, we start the loop again.

Since every time we project, it is onto a face of the Pareto frontier of dimension smaller than the current Pareto frontier, we can only take the loop at most $n$ times. After this we go to line~\ref{proc-mthc-9}, and since $\vect{x}$ is now a vertex of the Pareto frontier, we can use Lemma~\ref{lem:vert} to find whether there exists a strategy such that $\satisfies{s}{\MDP_i}\univ(\parobj)\wedge\bigwedge_{i = 1}^n \threshobj{\geq}{x_i}{\event \target_i}$.

If the output is \texttt{false}, we remark that since projection in line~\ref{proc-mthc-4} keeps all states belonging to strategies such that $\Pr_{\geq p}(\event F_i)$ in $\MDP_{i-1}$ (by Lemma~\ref{lem:mdp_proj_c}), it keeps all states belonging to strategies such that $\bigwedge_{i = 1}^n \threshobj{\geq}{x_i}{\event \target_i}$. Step~\ref{proc-mthc-5} may prune states used in strategies such that $\bigwedge_{i = 1}^n \threshobj{\geq}{x_i}{\event \target_i}$, but then by definition these strategies did not satisfy $\univ(\parobj)$. Hence, since none of the pruned states belonged to a strategies such that $\satisfies{s}{\MDP_i}\univ(\parobj)\wedge\bigwedge_{i = 1}^n \threshobj{\geq}{x_i}{\event \target_i}$, the algorithm is correct.
    
We show that \Cref{proc:mthc} solves at most a polynomial number of parity game of size polynomial in $|\Gamma|$. 

Indeed, to obtain $|\MDP_{i-1}^{\pi \vect{v}_i}|$, we have to compute the Pareto frontier of $|\MDP_{i-1}|$. This new MDP only has one more state than $\Gamma$ and at most $|\Gamma|$ additional transitions (that may lead from states originally in $\Gamma$ to the new state $\projSinkState$). 
We then remove once all states that do not satisfy $\univ(\parobj) \wedge \Pr_{= 1}(\event F)$, which can be done by solving a polynomial number of parity games. 
Every time we do this step, we project onto a face of the Pareto frontier of dimension smaller than the current Pareto frontier and this can only happen at most $n$ times. 
Thus we end up solving a number polynomial in $n$ of parity game of size polynomial in $\Gamma$, and compute Pareto frontiers for MDPs with $n$ objectives at most $n$ times.

We output \texttt{true} iff we find a strategy $\sigma$ that is a solution of $\satisfy{s}{\sigma}{\MDP_i}\univ(\parobj)\wedge \Pr_{= 1}(\event F)$ iff it is a solution of $\satisfies{s}{\MDP}\univ(\parobj)\wedge\limwedgeone{i \in [n]} \Pr_{\geq x_i}(\event F_i)$, and we can find strategies satisfying the left hand formula that use $2|\MDP||\parobj|$ memory, as proved in Theorem~\ref{thm:mth}.
\end{proof}

\begin{algorithm}[t]
    \caption{Sure parity and $n$ non-strict reachability threshold objectives}
    \label{proc:mthc}
    \textbf{Input:} MDP $\MDP$ -- clean w.r.t.\ $\parobj$ and $\target_1,\ldots,\target_n \subseteq \states$, a state $s \in \states$, a vector $\vect x \in [0,1]^n$ \\
    \textbf{Output:} A strategy $\strat$ such that $\satisfy{s}{\strat}{\MDP} \univ(\parobj) \wedge \bigwedge_{i = 1}^n \threshobj{\geq}{x_i}{\event \target_i}$ if it exists, else \texttt{false}.
    \begin{algorithmic}[1]
        \State\label{proc-mthc-1}Set $\MDP_0 = \MDP$ and $i=1$.
        \While {$\vect{x}$ is not a vertex of $P_i$, the Pareto frontier of $\MDP_i$} 
        \State\label{proc-mthc-3}Get $\vect{v}_i$ a vector orthogonal to the smallest face of $P_i$ that $\vect{x}$ belongs to.
        \State\label{proc-mthc-4}Compute $\MDP_{i}^{\pi \vect{v}_i}$ from $s$.  \Comment{By Lemma~\ref{lem:mdp_proj_c}.}
        \State\label{proc-mthc-5}Set $\MDP_i$ by taking $\MDP_{i-1}^{\pi \vect{v}_i}$ and pruning all states that do not satisfy $\univ(\parobj) \wedge \prob_{\geq 1}(\event F)$.
        \State\label{proc-mthc-6}If $\vect{x}$ is an interior point of $P_i$ return $\sigma$ s.t.\ $\satisfy{s}{\sigma}{\MDP_i}\univ(\parobj)\wedge\bigwedge_{i = 1}^n \threshobj{\geq}{x_i}{\event \target_i}$  \Comment{By Lemma~\ref{lem:int-par}.}  
        \State\label{proc-mthc-7} $i := i+1$
        \EndWhile
        \State\label{proc-mthc-9}Check if there exists $\sigma$ such that $\satisfies{s}{\MDP_i}\univ(\parobj)\wedge\bigwedge_{i = 1}^n \threshobj{\geq}{x_i}{\event \target_i}$. \Comment{By Lemma~\ref{lem:vert}.}
        \State\label{proc-mthc-10}If such a $\sigma$ does not exist then return \texttt{false}, else return $\sigma$.
    \end{algorithmic}
\end{algorithm}

\begin{example}
\label{ex:multi}
For the MDP from \Cref{fig:mdp-pareto}, with initial state $s$, and where $F_1$ is the yellow target and $F_2$ is the blue target, we check if $\satisfies{s}{\MDP}\univ(\parobj)\wedge\Pr_{\geq \nicefrac 2 3}(\event F_1)\wedge\Pr_{\geq \nicefrac 2 3 }(\event F_2)$. Point $(\nicefrac 2 3,\nicefrac  2 3)$ is on the Pareto frontier but not a vertex of it, and following Algorithm~\ref{proc:mthc}, a vector orthogonal to it is $(1,1)$. After projection on $(1,1)$, we obtain the MDP in Figure~\ref{fig:mdp-pareto-gen} (left), with both full and dashed transitions. 
As in Example~\ref{ex:proj}, no strategy satisfies surely the parity condition in the uppermost state, we prune the dashed transitions, restricting the Pareto frontier to between $(\nicefrac 1 3,1)$ and $(\nicefrac  23,\nicefrac 2 3)$. Now $(\nicefrac  2 3,\nicefrac 2 3)$ is a vertex of the new Pareto frontier, \Cref{lem:vert} tells us to project on vector $(2,1)$, and so we prune transition $\textcolor{red}{a}$ from the lowermost state, only keeping transition $\textcolor{green!50!black}{b}$. Since we can satisfy the parity condition from the lowermost state, we obtain that the strategy playing $\textcolor{green!50!black}{b}$ twice satisfies $\univ(\parobj)\wedge\Pr_{\geq \nicefrac 2 3}(\event F_1)\wedge\Pr_{\geq \nicefrac 2 3}(\event F_2)$.
\end{example}

\subsection{Proof of \Cref{thm:mult-th}}

We can finally prove Theorem~\ref{thm:mult-th} as follows:
We check if $\vect{x}$ is on the border of the Pareto frontier.
If it is not, \Cref{thm:approx-int} applies; otherwise, we apply \Cref{lem:mult-th}. 


\begin{proof}[of \Cref{thm:mult-th}]
    We check if $\vect{x}$ is on the border of the Pareto frontier of $\limwedgeone{i \in [n]} \Pr_{\geq x_i}(\event F_i)$. If it is not, $\satisfies{s}{\MDP}\univ(\parobj)\wedge\limwedgeone{i \in [n]} \Pr_{\geq x_i}(\event F_i)$ holds by Theorem~\ref{thm:approx-int}. Otherwise, it is on the border and we use Lemma~\ref{lem:mult-th} to decide if the property holds. Both cases give a finite-memory strategy $\sigma$ such that $\satisfy{s}{\sigma}{\MDP}\univ(\parobj)\wedge\limwedgeone{i \in [n]} \Pr_{\geq x_i}(\event F_i)$.
    
    We show that Algorithm~\ref{proc:mthc} solves at most a polynomial number of parity game of size polynomial in $|\Gamma|$. 
    Indeed, to obtain $|\MDP_{i-1}^{\pi \vect{v}_i}|$, we have to compute the Pareto frontier of $|\MDP_{i-1}|$. This new MDP only has one more state than $\Gamma$ and at most $|\Gamma|$ additional transitions (that may lead from states originally in $\Gamma$ to the new state $\projSinkState$). 
    We then remove once all states that do not satisfy $\univ(\parobj) \wedge \Pr_{= 1}(\event F)$, which can be done by solving a polynomial number of parity games. 
    Every time we do this step, we project onto a face of the Pareto frontier of dimension smaller than the current Pareto frontier and this can only happen at most $n$ times. 
    Thus we end up solving a number polynomial in $n$ of parity game of size polynomial in $\Gamma$, and compute Pareto frontiers for MDPs with $n$ objectives at most $n$ times. Our $\EXPTIME$ complexity comes from the Pareto frontier computations. Indeed, Theorem 2.1 of~\cite{etessami2007multi} shows that such Pareto frontiers may have more than polynomially many vertices (more precisely, they may have $|\MDP|^{\log(|\MDP|)}$ many vertices). Thus, to our knowledge, the way to compute this Pareto frontier is to evaluate all possible memoryless strategies, which is exponential in the size of $\MDP$.

    We output \texttt{true} iff we find a strategy $\sigma$ such that $\satisfy{s}{\sigma}{\MDP_i}\univ(\parobj)\wedge \Pr_{= 1}(\event F)$ iff $\sigma$ is a solution of $\satisfies{s}{\MDP}\univ(\parobj)\wedge\limwedgeone{i \in [n]} \Pr_{\geq x_i}(\event F_i)$. 

    The worst case memory bound comes from when we need to make a call to the strict threshold case, which as shown in \Cref{thm:approx-int} uses at most $2^{\poly(|\MDP| + nD)}$ memory.
\end{proof}
\section{Conclusion and Future Work}
\label{sec:ccl}
Combining sure parity and $n$ reachability threshold objectives can be done via a reduction to parity games in the case of strict thresholds and when maximizing the threshold lexicographically, and in exponential time with non-strict thresholds.
Finite-memory strategies suffice in all cases.
One direction for future work is to implement our algorithms in the probabilistic model checker Storm~\cite{DBLP:conf/cav/DehnertJK017}.
Further open problems include the case where targets are not sinks,
and the study of one sure parity and $n$ parity threshold objectives.
However, the exact memory required for one sure and \emph{one} almost-sure parity is already unknown.
It seems worthwhile to investigate if 1-bit Markov strategies suffice, as they do in countable MDPs with parity objectives~\cite{DBLP:conf/concur/KieferMST20}.
In~\cite{almagor2016minimizing,icalp2017}, the solution of sure parity and almost-sure reachability in MDPs relies on a reduction to a game with a \emph{fair} opponent.
Results in~\cite{DBLP:conf/cav/CastroDDP22} concern \emph{stochastic games} with a fair opponent, and may thus help extending the results from~\cite{almagor2016minimizing,icalp2017} to stochastic games.
Another possible extension is to consider combinations of multiple objectives in \emph{partially observable} MDPs (POMDPs), as in~\cite{DBLP:conf/aaai/Chatterjee0PRZ17}.

\clearpage
\newpage
\bibliography{biblio}

\end{document}